\documentclass[a4paper,UKenglish,cleveref,autoref,thm-restate,numberwithinsect]{lipics-v2021}

\title{The Complexity of the Distributed Constraint Satisfaction Problem}

\author{Silvia Butti}{Department of Information and Communication Technologies, Universitat Pompeu Fabra, Spain \and \url{https://sites.google.com/view/silviabutti/}}{silvia.butti@upf.edu}{https://orcid.org/0000-0002-0171-2021}{The project that gave rise to these results received the support of a fellowship from ``la Caixa'' Foundation (ID 100010434). The fellowship code is LCF/BQ/DI18/11660056. This project has received funding from the European Union’s Horizon 2020 research and innovation programme under the Marie Skłodowska-Curie grant agreement No. 713673.}

\author{Victor Dalmau}{Department of Information and Communication Technologies, Universitat Pompeu Fabra, Spain \and \url{https://www.upf.edu/web/victor-dalmau}}{victor.dalmau@upf.edu}{https://orcid.org/0000-0002-9365-7372}{Victor Dalmau was supported by MICCIN grants TIN2016-76573-C2-1P and PID2019-109137GB-C22.}

\authorrunning{S. Butti and V. Dalmau}

\Copyright{Silvia Butti and Victor Dalmau}

\ccsdesc{Theory of computation~Constraint and logic programming}

\keywords{Constraint Satisfaction Problems, Distributed Algorithms, Polymorphisms}

\category{} 

\relatedversion{}

\supplement{}

\acknowledgements{We would like to thank Gergely Neu for useful discussions on the weighted majority algorithm.}

\nolinenumbers 

\hideLIPIcs  

\EventEditors{Markus Bl\"{a}ser and Benjamin Monmege}
\EventNoEds{2}
\EventLongTitle{38th International Symposium on Theoretical Aspects of Computer Science (STACS 2021)}
\EventShortTitle{STACS 2021}
\EventAcronym{STACS}
\EventYear{2021}
\EventDate{March 16--19, 2021}
\EventLocation{Saarbr\"{u}cken, Germany (Virtual Conference)}
\EventLogo{}
\SeriesVolume{187}
\ArticleNo{44}

\usepackage[ruled]{algorithm2e}
\usepackage{wasysym}
\usepackage{tikz}

\newcommand{\mb}[1]{\mathbf{#1}}
\newcommand{\eps}{\varepsilon}
\newcommand{\oo}{\textbf{\textsf{O}}}
\newcommand{\pol}{\textnormal{\textsf{Pol}}}
\newcommand{\gm}{\Gamma}
\newcommand{\csp}{\textnormal{CSP}}
\newcommand{\cspgm}{\csp(\gm)}
\newcommand{\dcsp}{\textnormal{DCSP}}
\newcommand{\dcspgm}{\dcsp(\gm)}

\newcommand{\dsearch}{\dcsp\textnormal{-Search}}
\newcommand{\dser}{\dcsp\textnormal{-Search}(\gm)}
\newcommand{\mymod}[1]{\textnormal{ (mod {$#1$})}}

\newcommand{\BLP}{\textnormal{BLP}}
\newcommand{\eq}{\operatorname{eq}}
\newcommand{\Alg}{\operatorname{Alg}}
\newcommand{\nwmn}{\textnormal{`}\newmoon\textnormal{'}}
\newcommand{\bltr}{\textnormal{`}\blacktriangle\textnormal{'}}
\newcommand{\networks}[2]{\begin{tikzpicture}
    \foreach \n in {1,...,6}{
        \node[circle,fill=black] at ({\n*360/6}:#1cm) (n\n) {};
    }
    \draw (n1)--(n2);
    \draw (n2)--(n3);
    \draw (n3)--(n4);
    \draw (n4)--(n5);
    \draw (n5)--(n6);
    \draw (n6)--(n1);
    \draw (n1)--(n4);
    \draw (n2)--(n5);
    \draw (n3)--(n6);
    \begin{scope}[xshift=#2cm] 
    \foreach \n in {1,...,6}{
        \node[circle,fill=black] at ({\n*360/6}:#1cm) (n\n) {};
    }
    \draw (n1)--(n2);
    \draw (n2)--(n3);
    \draw (n3)--(n4);
    \draw (n4)--(n5);
    \draw (n5)--(n6);
    \draw (n6)--(n1);
    \draw (n1)--(n5);
    \draw (n2)--(n4);
    \draw (n3)--(n6);
    \end{scope}
\end{tikzpicture}}

\begin{document}
\maketitle
\begin{abstract}

We study the complexity of the Distributed Constraint Satisfaction Problem (DCSP) on a synchronous, anonymous network from a theoretical standpoint. In this setting, variables and constraints are controlled by agents which communicate with each other by sending messages through fixed communication channels. Our results endorse the well-known fact from classical CSPs that the complexity of fixed-template computational problems depends on the template's invariance under certain operations. Specifically, we show that DCSP($\Gamma$) is polynomial-time tractable if and only if $\Gamma$ is invariant under symmetric polymorphisms of all arities. Otherwise, there are no algorithms that solve DCSP($\Gamma$) in finite time. We also show that the same condition holds for the search variant of DCSP.

Collaterally, our results unveil a feature of the processes' neighbourhood in a distributed network, its iterated degree, which plays a major role in the analysis. We explore this notion establishing a tight connection with the basic linear programming relaxation of a CSP.

\end{abstract}

\section{Introduction}

The Constraint Satisfaction Problem (CSP) consists of a collection of variables and a collection of constraints where each constraint specifies the valid combinations of values that can be taken simultaneously by the variables in its scope. The goal is to decide if there exists an assignment of the elements of a domain to the variables which satisfies all constraints. The CSP is a very rich mathematical framework that is widely used both as a fruitful paradigm for theoretical research, and as a powerful tool for applications in AI, such as scheduling and planning \cite{Rossi2006, krokhin2017constraint}.

While, in its full generality, the finite-domain CSP is known to be NP-complete, applying specific restrictions on the instances can yield tractable subclasses of the problem. One of the most studied approaches consists in requiring that, in each constraint, the set of allowed combinations for its values be drawn from a prescribed set $\Gamma$, usually called the constraint language or the template. Thanks to the proof of the CSP dichotomy conjecture obtained separately in \cite{Bulatov2017} and \cite{Zhuk2017}, which culminated a decades-long research program, it is possible to determine the complexity (P or NP-complete) of each family of CSPs, $\cspgm$, which is obtained by fixing $\Gamma$. This proof confirmed that the complexity of the constraint satisfaction problem is deeply tied to certain algebraic properties of the constraint language. Specifically, it depends on whether or not the constraint language is invariant under certain operations known as its polymorphisms. The polymorphisms of a constraint language enforce a symmetry on the space of solutions of a CSP instance that can possibly be exploited by an algorithm. This connection with algebra is also present in our work.

We study the computational complexity of the distributed counterpart of CSP, which is known as DCSP. This was introduced by Yokoo et al. \cite{yokoo1992distributed} as a formal framework for the study of cooperative distributed problem solving. In particular, we consider a deterministic, synchronous, anonymous network of agents controlling variables and constraints, and we study the complexity of message passing algorithms on this network. A number of practical applications can be encoded in the DCSP model, for instance resource allocation tasks in wireless networks, routing, networking, and mobile technologies (see for instance \cite{duffy2013decentralized,bejar2001distributed}). 
  
We notice that this framework is general enough to encompass some simple Graph Neural Network architectures (see for example \cite{morris2019weisfeiler, grohe2020word2vec}). In particular, when training a GNN to classify graphs, it is customary that the GNN network ignores the node label when updating its feature vector. This is, in fact, essential as otherwise there would be no way to apply the network trained on a given graph to another one. However, whereas in all variants of GNNs the computation is limited to a reduced number of operations over feature vectors, in the $\dcsp$ model the computation at each node is governed by an arbitrary algorithm. GNNs have a wide range of applications including molecule classification or image classification (see \cite{GNNs2018} for example). Recently, GNNs have been deployed to solve CSPs \cite{GroheGNNCSP2019}.

While there are a variety of well-performing distributed algorithms for constraint satisfaction and optimisation (see for instance \cite{yokoo2000algorithms,meisels2008distributed,Fioretto2018}), the theoretical aspects of distributed complexity are to date not well understood. In this paper we initiate the study of the complexity of $\dcsp$ parametrized by the constraint language, obtaining a complete characterization of its tractable classes. More specifically, building on the connection between the CSP and algebra, we show that for any finite constraint language $\gm$, the decision problem for $\dcspgm$ is tractable whenever $\gm$ is invariant under symmetric polymorphisms of all arities, where an operation is symmetric if its result does not depend on the order of its arguments. Otherwise, there are no message passing algorithms that solve $\dcspgm$. Collaterally, we show that the same holds for the search problem for DCSP.

Our work begins with the identification of a feature of the nodes in a distributed network, its iterated degree, which plays a major role in how messages are transmitted in the network. The iterated degree is an extension of the similar concept introduced in the study of the isomorphism problem which turns out to have a variety of alternative characterizations in terms of fractional isomorphisms, the Weisfeiler-Leman test, and definability with counting logics (see \cite{grohe2020word2vec}). It turns out that, due to the network anonymity, in every distributed algorithm all equivalent agents (with respect to iterated degree) must necessarily behave identically at each round. A similar phenomenon has been observed independently in the context of GNNs in \cite{morris2019weisfeiler,XuHowPowerful} leading to further study in \cite{Barcelo2020logical}.

We use this fact to show that, under the absence of symmetric polymorphisms of any arity in $\Gamma$, it is always possible to construct two instances of $\dcspgm$, one satisfiable and the other unsatisfiable, that cannot be distinguished by any message passing algorithm in an anonymous network.

On the other hand, invariance under symmetric polymorphisms is connected with the basic linear programming relaxation of a CSP instance. More precisely, if $\Gamma$ has symmetric polymorphisms of all arities then one can decide the satisfiability of every instance of $\cspgm$ by checking whether its basic linear programming relaxation is feasible (see for instance \cite{Barto2017}). Whereas it is not clear how to directly use this fact to obtain a distributed algorithm for $\dcsp(\Gamma)$, it can be applied to establish a structure theorem that unveils a simple yet surprising structure in the solution space of every satisfiable instance in $\dcsp(\Gamma)$: it must contain a solution that assigns the same value to all variables that have the same iterated degree. The proof of the structure theorem uses the weighted majority algorithm, a weight update method that is widely used in optimisation and machine learning applications (see \cite{Arora2012}). The structure theorem is key in the proof of the positive results as it allows to run an adapted variant of the $jpq$-consistency algorithm \cite{kozik2020solving} that overcomes the absence of unique identifiers for the variables, by using instead their iterated degree.

This paper is organised as follows. In Section \ref{sec:prelims} we introduce some definitions and technical concepts about the DCSP model. In Section \ref{sec:BLP} we present the basic LP relaxation for CSPs and we show its connection to the symmetry on the solution space, culminating in the statement of the structure theorem. Section \ref{sec:complexity} is dedicated to the proof of the dichotomy theorem for the complexity of $\dcsp$, with the hardness results in Section \ref{sec:intractable}, the details of the distributed algorithm for tractable languages in Section \ref{sec:tractable}, and its extension to the search problem in Section \ref{sec:search}. In the Conclusion we discuss some directions into which our work could be extended. Finally, in the Appendix we add some technicalities and provide detailed proofs for all the claims that were made along the paper.

\section{Preliminaries}\label{sec:prelims}

\paragraph*{Constraint Satisfaction Problems.}
An instance $I$ of the finite-domain \textit{Constraint Satisfaction Problem} (CSP) is a triple $ (X,D,C)$ where $X$ is a set of variables, $D$ is a finite set called the domain, and $C$ is a set of constraints where a constraint $c \in C$ is a pair $(\mb{s}, R)$ where $R \subseteq D^{k}$ for $k$ a positive integer, $R$ is a relation over $D$ of arity $k$, and $\mb{s}$ is a tuple of $k$ variables, known as the \textit{scope} of $c$. We use $arity(\cdot)$ to denote the arity of a relation, tuple, or constraint and we write $x \in c$ for any variable $x$ in the scope of $c$. An \textit{assignment} $\nu: X \to D$ is said to be \textit{satisfying} if for all constraints $c = ( \mb{s}, R ) \in C$ we have $\nu(\mb{s}) \in R$, where $\nu$ is applied to $\mb{s}$ coordinate-wise. Usually we denote the number of variables by $n$ and the number of constraints by $m$. 

Let $\gm$ be a set of relations over some finite domain $D$, and let $\cspgm$ denote the set of CSP instances with all constraint relations lying in $\gm$. In this context, $\gm$ is known as the \textit{constraint language}. Throughout this paper, we will assume that $\gm$ is always finite. Then, the \textit{decision problem} for $\cspgm$ is the problem of deciding whether a satisfying assignment exists for an instance $I \in \cspgm$. The \textit{search problem} for $\csp(\gm)$ is the problem of deciding whether a satisfying assignment exists and, if it does, to find one such assignment. 

\paragraph*{The Distributed Model.}
We consider the DCSP model of \cite{yokoo1992distributed} with some small modifications. The basic idea is to assign the task of solving a constraint satisfaction problem to a multi-agent system. In the original model, which assumes that all constraints are binary \cite{yokoo1998distributed, yokoo2000algorithms}, the assumption is that each variable is controlled by an agent, and two agents can communicate with one another if and only if they share a constraint. Here we deviate slightly from the original model to allow for non-binary constraints and we assume that both variables and constraints are controlled by distributed agents in the network. An instance of the \textit{Distributed Constraint Satisfaction Problem} (DCSP) is a tuple $(A, X, D, C, \alpha )$, where $X$, $D$, and $C$ are as in the classical CSP, $A$ is a finite set of agents, and $\alpha: X \cup C \to A$ is a surjective function which assigns the control of each variable $x \in X$ and each constraint $c \in C$ to an agent $\alpha(x)$, $\alpha(c)$ respectively. For the purpose of this paper, we assume that there are exactly $n+m$ agents, and therefore each agent controls exactly one variable or one constraint. This can be done without loss of generality since any agent controlling multiple nodes can simulate multiple agents, each controlling a node. Under this assumption, there is a one-to-one correspondence between instances of CSP and DCSP, and thus we shall switch freely between them.

\paragraph*{Distributed Networks and Message Passing.} We now present some fundamental concepts relating to the message-passing paradigm for distributed networks. For a general introduction to distributed algorithms, we refer the reader to \cite{fokkink2013distributed}. A \textit{distributed system} consists of a finite set of nodes or processes, which are connected through communication channels to form a network. Any process in the network can perform events of three kinds: \textit{send}, \textit{receive} and \textit{internal}. Send and receive events are self-explanatory, as they denote the sending or receiving of a message over a communication channel. Any kind of local computation performed at the process level, as well as state changes and decisions, are classified as internal events.

We assume a fully \textit{synchronous} communication model, meaning that the send event at a process $a$ and the corresponding receive event at a process $a'$ can be considered \textit{de facto} as a unique event, with no time delay. As a whole, a synchronous system proceeds in rounds, where at each round a process can perform some internal computation and then send messages to and receive messages from its neighbours. A round needs to terminate at every process before the next round begins. Note that while for simplicity we assume a synchronous network, all our algorithms can be adapted to asynchronous systems by applying a simple synchronizer. Nonetheless, we point out that our negative results rely on the network operating in synchronous rounds.

We make the fundamental assumption that the network is \textit{anonymous}, meaning that variables, constraints and agents do not have IDs. For practical purposes, we still refer to variables and constraints with names (such as $x_{i}$, $c_{i}$), however these cannot be communicated through the channels. The assumption of anonymity can have various practical justifications: the processes may actually lack the hardware to have an ID, or they may be unable to reveal their ID due to security or privacy concerns. For instance, the basic architecture of GNNs is anonymous. This is a very desirable property as it allows to deploy GNNs in different networks than those in which they were trained.

We assume that all the processes run locally the same deterministic algorithm, therefore IDs cannot be created and deadlocks cannot be broken by for instance flipping a random coin. Hence, the lack of IDs makes the processes essentially indistinguishable from one another - except, as we will see later, for the structure of their neighbourhood in the network.

Leader election is a procedure by which the processes in a network select a single process to be the leader in a distributed way. If a leader is elected, then she can assign unique identifiers to every process. Moreover, all the information about the instance can be gathered to the leader, who can then solve the CSP locally. It is a well-known result that there does not exist a terminating deterministic algorithm for electing a leader in an anonymous ring \cite{Alguin1980Local}. Therefore, the assumptions of anonymity and determinism ensure that the DCSP model is intrinsically different from the (centralised) CSP framework, and open up the way for establishing novel, non-trivial complexity results. We remark that while considerable effort has been put into characterizing under what conditions an anonymous network is able to elect a leader \cite{Boldi1996Symmetry, Yamashita1988} or compute relations \cite{Boldi2001effective}, our work focuses on characterizing the complexity of the DCSP as parametrised by the constraint language. Therefore, all of our algorithms work regardless of the topology of the network, and hence regardless of whether or not a leader can be elected.

The encoding of a DCSP instance into the message passing framework is straightforward. The processes correspond to the agents of the network, and there is a labelled communication channel between a variable agent $\alpha(x)$ and a constraint agent $\alpha(c)$ if and only if $x \in c$. More formally, the \textit{Factor Graph} \cite{Fioretto2018} $G_{I}$ of an instance $I= (X,D,C)$ of $\csp$ is the undirected bipartite graph with vertex set $X \cup C$ 
and edge set $\{\{x,c\} \mid x \in c\}$. Each edge in $G_{I}$ that is incident to a variable $x$ and a constraint $c$ where $c=( \mb{s},R )$ has a label $\ell_{x,c}=(S,R)$ for $S=\{i \mid \mb{s}[i]=x\}$, where for a tuple $\mb{t}$, $\mb{t}[i]$ denotes the $i^{\textnormal{th}}$ entry of $\mb{t}$.\footnote{For mathematical clarity, we label edges with the relation itself. However, in algorithmic applications, every relation can be substituted with a corresponding symbol.} Then, the message passing network corresponds to the factor graph where every node (variable or constraint) is replaced by their associated agent and every edge by a communication channel of the same label. Note that between any two nodes there is at most one channel. If privacy is a concern, we point out that labeling channels does not reveal any more information about the processes than what is strictly necessary for the problem instance to be well defined. Unless explicitly stated we only consider instances whose factor graph consists of a unique connected component. It is easy to prove (see the Remark \ref{re:alternativemodel} in the Appendix) that in the case that all relations are binary, the original model where only variables are controlled by agents is equivalent to our model. 
 
At the start of an algorithm, a process only has access to very limited information. All processes know the total number $n$ of variables in the CSP instance, the total number $m$ of constraints, the labels of the communication channels that they are incident to in the network, and naturally whether they are controlling a variable or a constraint. During a run of the algorithm a process can acquire further knowledge from the messages that it receives from its neighbours. We assume that at any time each process is in one of a set of states, a subset of which are terminating states. When it enters a terminating state, a process performs no more send or internal events, and all receive events are disregarded. The local algorithm is then a deterministic function which determines the process' next state, and the messages it will send to its neighbours. The output of such function only depends on the process' current knowledge, on its state, and on the global time. We allow processes to send different messages through different channels. However, since processes can only distinguish the channels based on their labels, identical messages must be sent through channels with identical labels. Note that the power of the model would not decrease if only one message was allowed to be passed through all the channels, since a process can simulate sending a separate message through each channel by tagging each message with the label of the desired channel and concatenating them in a unique string. This, however, comes at the cost of increased message size. Moreover, if a process needs to broadcast multiple messages, these can be concatenated into one. We say that an algorithm terminates when all processes are in a terminating state.

We say that a distributed algorithm solves an instance $I$ of $\dcsp$ if the algorithm terminates and the terminating state of every process correctly states that $I$ is satisfiable if it is, and that it is not satisfiable otherwise. Moreover, we consider the search version of $\dcsp$, denoted 
$\dsearch$. In the search version, if the input instance $I$ is satisfiable, the terminating state of every variable process $\alpha(x)$ must additionally specify a value $\nu(x)\in D$ such that $\nu: X \to D$ is a satisfying assignment. For every constraint language $\gm$, we denote by $\dcspgm$ and $\dser$  the restrictions of $\dcsp$ and $\dsearch$, respectively, to instances containing only constraint relations from $\Gamma$.

In terms of algorithmic complexity, there are a number of measures that can be of interest. Time complexity, which is our primary concern, corresponds to the total amount of time required for the algorithm to terminate, including the time needed for internal events. This is closely related to the number of rounds of the algorithm, which is another measure that we are concerned with. Message complexity and bit complexity measure the total number of messages and bits exchanged respectively. These can be bounded easily from the maximum size of a message.

\paragraph*{Iterated Degree and Degree Sequence.}
We present a number of concepts from graph theory that carry over to CSPs. Their adaptation to DCSPs is straightforward in all cases. In an undirected graph $G$, the degree of a vertex $v$ is the number of edges incident at $v$. The zeroth iterated degree of $v$ is equal to its degree. For $k \geq 1$, the $k^{th}$ iterated degree of $v$ is the multiset of $(k-1)^{th}$ degrees of $v$'s neighbours in $G$. The $k^{th}$ iterated degree sequence of a graph is the multiset of $k^{th}$ iterated degrees of its vertices.

\begin{example}
In the context of graph theory the \textit{colour refinement algorithm}, which calculates the iterated degree sequence of a graph, is often used as a simple heuristic for the graph isomorphism problem. If two graphs are isomorphic then they must have the same iterated degree sequence, but the opposite is not true (see for example Figure \ref{fig:networks}).\lipicsEnd

\end{example}
\begin{figure}[t]
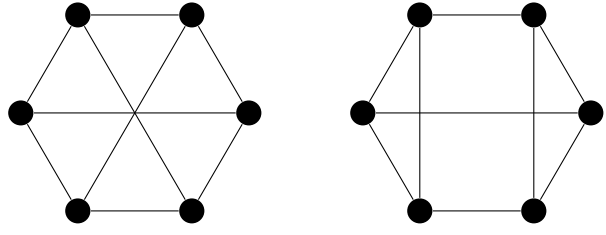
 
    \centering
    \networks{1.5}{4.5} 
    \caption{Both graphs depicted above are 3-regular and hence they have the same iterated degree sequence. However, they are clearly not isomorphic, since the left graph is bipartite while the right one is not.}
    \label{fig:networks}
\end{figure}

We extend the notion of iterated degree to CSPs as follows. Consider the labelled factor graph $G_{I}$ of an instance $I$ described in the previous paragraph. In what follows it will be convenient to allow instances $I$ with a disconnected factor graph $G_{I}$. Let $v$ be a node of $G_{I}$ and denote its neighbourhood in the factor graph by $N(v)$. The (zeroth) \textit{degree}, denoted $\delta_{0}(v)$, of a node in the factor graph is simply a symbol that distinguishes variables from constraints: we set $\delta_{0}(x)= \nwmn$ for all $x \in X$ and $\delta_{0}(c)=\bltr$ for all $c \in C$. The \textit{$k^{th}$ iterated degree}\footnote{We remark that the notions of degree and iterated degree are well-defined concepts in graph theory. We borrow this terminology to refer to the analogous concepts in CSPs.} ($k \geq 1$) of a node $v$ is defined as $\delta_{k}(v) =\{(\ell_{v,w}, \delta_{k-1}(w)) \mid w \in N(v)\}$. We write $v \sim^{k}_{\delta} v'$ if $\delta_k(v)=\delta_k(v')$, and simply $v \sim_{\delta} v'$ if $v \sim^{k}_{\delta} v'$ for all $k\geq 0$. In this case, we say that $v$ and $v'$ are \textit{iterated degree equivalent}. It can be shown (see Proposition \ref{prop:2nisEnough}) that as $k$ increases, the partition induced by $\sim_{\delta}^{k}$ gets more refined, and indeed it reaches a fixed point for some $k \leq 2n$ where $n=|X|$. The notion of iterated degree is strikingly relevant in our work as it captures what it means for two processes in a network to be indistinguishable. This implies that no distributed algorithm can differentiate between two iterated degree equivalent nodes, as we illustrate in the following result.

\begin{restatable}{proposition}{idsisameSol}
\label{prop:idsi->sameSol}
   Let $I=(A,X,D,C,\alpha)$ be an instance of $\dcspgm$ whose factor graph is not necessarily connected and consider two variables $v,v' \in G_{I}$. Then, $v \sim_{\delta} v'$ if and only if any terminating decision algorithm over $I$ outputs the same decision at $\alpha(v)$ and $\alpha(v')$. Furthermore, if $v,v'\in X$ and $I$ is satisfiable, then any terminating search algorithm outputs the same values $\nu(v) = \nu(v')$ at $\alpha(v)$ and $\alpha(v')$.
\end{restatable}

The following is a direct consequence of Proposition \ref{prop:idsi->sameSol}. We say that two instances $I$ and $I'$ have the same \emph{iterated degree sequence} if there exists a bijection $\gamma$ between the nodes of $G_{I}$ and the nodes of $G_{I'}$ such that for every $k\geq 0$ and every node $v$ of $G_{I}$, the $k^{th}$ degree of $v$ in $I$ is equal to the $k^{th}$ degree of $\gamma(v)$ in $I'$. We note that in this case, if we construct the (disconnected) instance $I\cup I'$ containing all the variables and constraints in $I$ and $I'$, then $v\sim_{\delta} \gamma(v)$ for every node $v\in G_{I}$. Hence the result below follows.

\begin{corollary} \label{cor:sameDegSeq->sameSol}
Let $I,I' \in \dcspgm$ have the same iterated degree sequence. Then with both inputs any terminating decision algorithm will report the same decision.
\end{corollary}

\paragraph*{Polymorphisms.} Let $R$ be a $k$-ary relation over a finite domain $D$. An $\ell$-ary \textit{polymorphism} of $R$ is an operation $f: D^{\ell} \to D$ such that the coordinate-wise application of $f$ to any set of $\ell$ tuples from $R$ gives a tuple in $R$. More precisely, for any $\mb{t}_{1}, \ldots, \mb{t}_{\ell} \in R$, we have that $(f(\mb{t}_{1}[1],\ldots,\mb{t}_{\ell}[1]), \ldots, f(\mb{t}_{1}[k],\ldots,\mb{t}_{\ell}[k])) \in R$.
We say that a function $f$ is a polymorphism of a constraint language $\gm$ if $f$ is a polymorphism of all relations $R \in \gm$. Equivalently, we say that $\gm$ is \textit{invariant} under $f$. The set of polymorphisms of a constraint language $\gm$ will be denoted by $\pol(\gm)$. There is a particular construction of a CSP instance that is closely related to the clone of polymorphisms of the corresponding constraint language. Let $\Gamma$ be a constraint language over a finite domain $D$. For any positive integer $r$, the \textit{indicator problem of order $r$} for $\Gamma$ is the instance
$I = (X,D,C) \in \cspgm$ where $X = D^r$ and $C$ contains for every relation $R\in \gm$ and for every $\mb{t}_1,\dots,\mb{t}_r\in R$, the constraint $(\mb{s},R)$ where $\mb{s}[i]=(\mb{t}_1[i],\dots,\mb{t}_r[i])$ for every $i\in\{1,\dots,arity(R)\}$. It follows easily that for every $\nu:D^r\rightarrow D$, 
$\nu$ satisfies $I$ if and only if $\nu$ is a polymorphism of $\Gamma$.

An $\ell$-ary operation $f$ is said to be \textit{symmetric} if for all $x_{1}, \ldots, x_{\ell}$ and for all permutations $\sigma$ of $\{1,\dots,\ell\}$ we have that $f(x_{1}, \ldots, x_{\ell})=f(x_{\sigma(1)}, \ldots, x_{\sigma(\ell)})$.

\begin{example}
Consider the Boolean relation $R=\{(0,1),(1,0)\}$. 
It is easy to see that the ternary minority operation $f$ given by $f(x,y,z) = x \oplus y \oplus z $ is a polymorphism of $R$. On the other hand, one can show that $R$ does not have symmetric polymorphisms of arity 2. In particular, let $\textbf{t}_1=(0,1)$ and $\textbf{t}_2=(1,0)$. Since a symmetric binary operation $f$ needs to satisfy $f(0,1)=f(1,0)$, the coordinate-wise application of $f$ to $\textbf{t}_1, \textbf{t}_2$ would yield a reflexive tuple, which cannot possibly belong to $R$.\lipicsEnd
\end{example}
Our work unveils a novel structure in the space of solutions of a CSP instance that is deeply connected to the symmetry of its polymorphisms. In particular, $\pol(\gm)$ containing symmetric polymorphisms of all arities is equivalent to the existence of a satisfying assignment to every satisfiable instance of $\cspgm$ that preserves the partition induced by $\sim_{\delta}$. This is the main result of the next section.

\section{Basic Linear Programming relaxation} \label{sec:BLP}

For any CSP instance $I = (X,D,C)$ there is a LP relaxation (usually called \emph{basic LP relaxation}, see for example \cite{Kun2012}) denoted $\BLP(I)$, which is defined as follows. It has a variable $v(x,d)$ for each $x \in X$ and $d \in D$, and a variable $v(c,\mb{t})$ for each $c\in C$ and $\mb{t}\in R$ where $R$ is the constraint relation of $c$. All variables must take values in the range $[0,1]$. The value of $v(x,d)$ is interpreted as the probability that $v$ is assigned to $d$. Similarly, the value of $v(c,\mb{t})$ is interpreted as the probability that the scope of $c$ is assigned component-wise to the tuple $\mb{t}$. In this paper we only deal with a feasibility problem (that is, there is no objective function). The variables are restricted by the following equations:
\begin{align}
  \sum_{d \in D}v(x,d) &= 1 \quad \textnormal{ for all }x \in X \label{eq:LP2}\\
  \sum_{\substack{\mb{t}\in R_c\\ \mb{t}[i]=d}} v(c,\mb{t}) &- v(\mb{s}_c[i],d) = 0\quad \textnormal{ for all }c \in C \textnormal{, all }i\in \{1,\dots,arity(c)\} \textnormal{, and all } d\in D \label{eq:LP3}
\end{align}
\noindent
where we denote the relation and scope of a constraint $c$ by $R_{c}$ and $\mb{s}_{c}$ respectively. We say that $\BLP$ decides $\cspgm$ if for every instance $I \in \cspgm$, $I$ is satisfiable whenever $\BLP(I)$ is feasible. We will use the following well-known result, which for the reader's convenience we prove in the Appendix.

\begin{restatable}[see \cite{Kun2012}]{theorem}{thmrounding} \label{thm:rounding} 
If $\gm$ has symmetric polymorphisms of all arities, then BLP decides $\cspgm$. Moreover, if $I \in \cspgm$ is satisfiable then it has a solution $\nu$ such that for all $x,x'$ with $v(x,d) = v(x',d)$ for all $d \in D$, we have $\nu(x)=\nu(x')$.
\end{restatable}

The following theorem reveals a useful structure inside the solutions of the BLP.

\begin{restatable}{theorem}{thLP}
\label{th:LP}
Let $I=(X,D,C)$ be an instance of $\csp(\gm)$ such that $\BLP(I)$ is feasible. Then, $\BLP(I)$ has a feasible solution such that for every $x,x'\in X$ with $x\sim_{\delta} x'$ and every $d\in D$, $v(x,d)=v(x',d)$.
\end{restatable}

\begin{proof}[Proof (Sketch).]
We start by rewriting the program in the form
\begin{equation} \label{eq:feasibProblem}
    \exists? \mb{v} \in [0,1]^{V} \quad B\mb{v} \geq \mb{b}
\end{equation}
by replacing every equality $a=b$ by the inequalities $a \geq b$ and $-a \geq -b$.

Let us use $W$ and $V$ to denote the rows and columns of $B$ respectively. The main idea of the proof is to apply the Multiplicative Weight Update (MWU) algorithm, a well-known technique that is widely used in optimisation and machine learning. MWU was discovered independently by researchers of different communities; for a survey of its different variants we refer the reader to \cite{Arora2012}. The version that is relevant to our paper is described in Algorithm \ref{alg:MWU}. Assuming that a feasible solution to (\ref{eq:feasibProblem}) does exist, the algorithm only requires the existence of an oracle which, given a probability $W$-vector $\mb{p}$
{(i.e, a non-negative vector $\mb{p}$ such that the sum of all its entries is $1$)}, outputs a vector $\mb{v}$ which is a solution to the weaker problem
\begin{equation} \label{eq:probFeasibProblem}
    \exists? \mb{v} \in [0,1]^{V} \quad \mb{p}^{T}B\mb{v} \geq \mb{p}^{T}\mb{b}
\end{equation}
\noindent
if one exists, or correctly states that no such vectors exist otherwise.

\begin{algorithm}[t]
\SetAlgoLined
\textbf{Initialisation:} Fix $\eta \leq \frac{1}{2}$, let $\rho = \max_{[0,1]^{V}} \max_{w \in W} |B_{w} \mb{v} -\mb{b}[w]|$, and let $\mb{w}^{(1)}$ be a $W$-vector, whose entries, called weights, are initially set to $1$.\\
\For{$t = 1, \ldots, T$ }{
Compute the probability vector $\mb{p}^{(t)} = \frac{1}{\Phi(t)} \mb{w}^{(t)}$, where $\Phi(t) = \sum_{j=1}^{|W|} \mb{w}^{(t)}[j]$\\
Let $\mb{v}^{(t)}$ be a solution satisfying $(\mb{p}^{(t)})^{T}B\mb{v}^{(t)} \geq (\mb{p}^{(t)})^{T}\mb{b}$ given by oracle $\oo$\\
 Compute the losses $\pmb{\ell}^{(t)} = \frac{1}{\rho} (B \mb{v}^{(t)}- \mb{b})$ \\
Compute the new weights $\mb{w}^{(t+1)} = \mb{w}^{(t)}(1-\eta \pmb{\ell}^{(t)}$) }
\KwRet $\mb{v}:= \frac{1}{T} \sum_{t=1}^{T} \mb{v}^{(t)}$\\
 \caption{Multiplicative Weight Update}\label{alg:MWU}
\end{algorithm}

Under some technical conditions that provide an upper bound on the number of rounds $T$ necessary to achieve a given approximation (see Theorem \ref{thm:oracle->MWU} in the Appendix) it follows that when $T\rightarrow\infty$ MWU converges to a solution of $\BLP(I)$. Now consider an oracle $\oo$ that, given a $W$-vector $\mb{p}$, returns the $V$-vector $\mb{v}$ where for every $v\in V$, $\mb{v}[v]=1$ if $\mb{p}^T B[v]$ is positive and $\mb{v}[v]=0$ otherwise. Since $\mb{v}$ maximizes $\mb{p}^T B\mb{v}$ under the restriction $\mb{v}\in[0,1]^V$ it follows that 
$\mb{v}$ satisfies (\ref{eq:probFeasibProblem}).

We note that $\sim_{\delta}$ induces an equivalence relation on the variables of $\BLP(I)$ (namely, $v(x,d)$ is equivalent to $v(x',d')$ whenever $x\sim_{\delta} x'$ and $d=d'$) which can be extended to an equivalence relation $\sim_V$ on the set $V$ of columns in $B$. Similarly, $\sim_{\delta}$ induces an equivalence relation $\sim_W$ on the rows $W$ of $B$ in a natural way. Then our goal is to show that the positions of $\sim_V$-equivalent entries in the output $\mb{v}:= \frac{1}{T} \sum_{t=1}^{T} \mb{v}^{(t)}$ are identical. This is done by showing by induction the more general fact that at each iteration $t$ of the algorithm, the positions of all $\sim_V$-equivalent entries in $\mb{v}^{(t)}$ are identical, and that for each of the $W$-vectors ($\mb{w}^{(t)}$, $\mb{p}^{(t)}$, and $\pmb{\ell}^{(t)}$) the positions of all $\sim_W$-equivalent entries are identical as well.
\end{proof}
We finalize the section by presenting the theorem on the structure of the solution space of CSP instances. 
\begin{theorem} \label{thm:sympol=iso=idsi}
Let $\gm$ be a finite constraint language. The following are equivalent:
\begin{enumerate}
    \item $\gm$ has symmetric polymorphisms of all arities.
    \item For all satisfiable instances $I=(X,D,C) \in \cspgm$ there exists a satisfying assignment $\nu: X \to D$ such that for all pairs of variables $x,x' \in X$, if $x \sim_{\delta} x'$ then $\nu(x)=\nu(x')$.
\end{enumerate}
\end{theorem}

\begin{proof}
$(1)\Rightarrow (2)$. Let $I$ be a satisfiable instance of $\cspgm$, where $\gm$ has symmetric polymorphisms of all arities. Consider the solution of $\BLP(I)$ given by Theorem \ref{th:LP} and note that it satisfies $v(x,d)=v(x',d)$ for all $x \sim_{\delta} x'$ and all $d \in D$. Then, by Theorem \ref{thm:rounding}, $I$ has a solution $\nu$ which satisfies $\nu(x)=\nu(x')$ for all $x \sim_{\delta} x'$.

$(2)\Rightarrow (1)$. Let $\gm$ satisfy $(2)$ and let $r \geq 1$. We shall prove that $\gm$ has a symmetric polymorphism of arity $r$. Let $I= (X,D,C)$ be the indicator problem of order $r$. Recall that every solution to $I$ corresponds to an $r$-ary polymorphism of $\gm$, and hence the indicator problem is always satisfiable since for instance the projection to the first coordinate is a polymorphism of $\gm$. Let $\nu$ be a solution of the indicator problem which satisfies condition (2). 
It is easy to show by induction that for every tuple $(t_1.\dots,t_r)\in D^r$, every permutation $\sigma$ of $\{1,\dots,r\}$ and every 
$k\geq 0$, $(t_1,\dots,t_r)\sim_{\sigma}^k (t_{\sigma(1)},\dots,t_{\sigma(r)})$ which implies that $\nu(t_1,\dots,t_r)=\nu(t_{\sigma(1)},\dots,t_{\sigma(r)})$.
We conclude that $\nu$ is symmetric as required.
\end{proof}

\section{The Complexity of DCSP} \label{sec:complexity}
The primary goal of this section is to prove the main theorem of this paper, namely, the dichotomy theorem for tractability of $\dcspgm$, which we now state.

\begin{theorem} \label{thm:maintheorem-dec}
$\dcspgm$ is solvable in polynomial time if and only if $\pol(\gm)$ contains symmetric polymorphisms of all arities. Otherwise, $\dcspgm$ cannot be solved in finite time.
\end{theorem}

We show hardness of constraint languages that do not have symmetric polymorphisms of all arities in Section \ref{sec:intractable} and tractability of the remaining languages in Section \ref{sec:tractable}. In addition, in Section \ref{sec:search} we extend the decision algorithm so that, additionally, it also provides a solution to the search problem. Hence we have:

\begin{theorem}
\label{thm:maintheorem-ser}
$\dser$ is solvable in polynomial time if and only if $\pol(\gm)$ contains symmetric polymorphisms of all arities. Otherwise, $\dser$ cannot be solved in finite time.
\end{theorem}

\subsection{Intractable Languages} \label{sec:intractable}

In this section we focus on intractable languages, that is, the hardness part of Theorem \ref{thm:maintheorem-dec}.

\begin{theorem} \label{thm:nopols->intractable}
Let $\gm$ be a constraint language on a finite domain $D$. If $\pol(\gm)$ does not contain symmetric operations of all arities, then there is no algorithm that solves $\dcspgm$ in finite time.
\end{theorem}

Schematically, the proof goes as follows. Assume that $\gm$ does not have symmetric polymorphisms of some arity $r$. Consider the relation $U$ defined by the set of solutions of the indicator problem of order $r$. 
It can be shown that if $\dcspgm$ is solvable in polynomial (or finite) time then so is $\dcsp(\{U\})$. Then, we show that there always exist two instances of $\dcsp(\{U\})$, one which is satisfiable and the other one which is not, that have the same iterated degree sequence. Therefore, any algorithm will return the same output on both instances, meaning that one of these outputs is wrong. Before embarking on the proof we state the following useful combinatorial lemma.

\begin{restatable}{lemma}{letechnical}
\label{le:technical}
Let $0<k<d$ be positive integers. 
If $n$ is a large enough multiple of $k$, then there exists a collection $\mathbb{S}$ of $n^{k}$ $k$-element subsets of $\{0,1,\dots,kn-1\}$ satisfying the following properties:
\begin{enumerate}
    \item[\textnormal{(a)}] $\mathbb{S}$ contains every $k$-element subset of $\{0,\dots,d-1\}$
    \item[\textnormal{(b)}] Every element of $\{0,1,\dots,kn-1\}$ appears in the same number of sets of $\mathbb{S}$.
\end{enumerate}
\end{restatable}

\begin{proof}[Proof of Theorem \ref{thm:nopols->intractable}]
Assume that $\pol(\gm)$ does not contain symmetric polymorphisms of arity $r$. Fix any arbitrary order $\mb{t}_{1}, \ldots, \mb{t}_{|D|^{r}}$ on the tuples of $D^r$ and consider the relation $U$ defined as 
\[\{(f(\mb{t}_{1}), \ldots, f(\mb{t}_{|D|^{r}})) \mid f \textnormal{ is a polymorphism of $\Gamma$ of arity $r$}\}.\]
This is, $U$ is the set of solutions of the indicator problem of order $r$. It follows easily (see Remark \ref{re:indicator} in Appendix) that if $\dcsp(\{U\})$ is not solvable in finite time then neither is $\dcspgm$. In particular, this follows from an adaptation of standard complexity reductions, given that $U$ is pp-definable from $\gm$ without using equality.

Partition $D^r$ into equivalence classes where two tuples $\mb{t},\mb{t}'\in D^r$ are related, denoted $\mb{t}\equiv\mb{t}'$, if there exists some permutation $\sigma$ on $\{1, \ldots, r\}$ such that $\mb{t}'[i]=\mb{t}[\sigma(i)]$ for every $i \in \{1, \ldots, r\}$. We shall use $D^r_{\equiv}$ to refer to the collection of classes and $[\mb{t}]_{\equiv}$ to refer to the class of tuple $\mb{t}$. For every $\mb{t} \in D^r$, define $k_{[\mb{t}]_{\equiv}}$ to be the number of tuples in $[\mb{t}]_{\equiv}$. Then we can choose an integer $n$ large enough such that
for every $\mb{t}\in D^r$, $n$ is a multiple of $k_{[\mb{t}]_{\equiv}}$, and $n$ satisfies 
Lemma \ref{le:technical} for $k=k_{[\mb{t}]_{\equiv}}$ and $d=k_{[\mb{t}]_{\equiv}} \cdot |D|$.

We are now ready to construct two instances $I_{1}$ and $I_{2}$ of $\dcsp(\{U\})$, which are indistinguishable with respect to their iterated degree sequence, but differ with regards to satisfiability. The two instances have the same set of variables, defined to be $\bigcup_{[\mb{t}]_{\equiv}\in D^r_{\equiv}} V_{[\mb{t}]_{\equiv}}$ where 
$V_{[\mb{t}]_{\equiv}}=\{v_{[\mb{t}]_{\equiv}}^i \mid 0\leq i< k_{[\mb{t}]_{\equiv}} n\}$ is a set of $k_{[\mb{t}]_{\equiv}}n$ distinct variables.

We start by constructing the constraints of the unsatisfiable instance $I_{1}$, which we will do in two stages. First, for every class $[\mb{t}]_{\equiv}$, let $\mathbb{S}_{[\mb{t}]_{\equiv}}$ be the collection of $n^{k_{[\mb{t}]_{\equiv}}}$ sets of cardinality $k_{[\mb{t}]_{\equiv}}$ given by Lemma \ref{le:technical}, as before with $d=k_{[\mb{t}]_{\equiv}}\cdot |D|$ and $k=k_{[\mb{t}]_{\equiv}}$. Note that each set in $\mathbb{S}_{[\mb{t}]_{\equiv}}$ defines naturally a subset of $V_{[\mb{t}]_{\equiv}}$ so we shall abuse notation and assume
that $\mathbb{S}_{[\mb{t}]_{\equiv}}$ is a collection of subsets of $V_{[\mb{t}]_{\equiv}}$.

To simplify notation it will be convenient to use $\mathbb{S}$ as a shorthand for the indexed family $\{\mathbb{S}_{[\mb{t}]_{\equiv}} \mid {[\mb{t}]_{\equiv}}\in D^r_{\equiv}\}$. Now let $S$ be $\{S_{[\mb{t}]_{\equiv}} \mid {[\mb{t}]_{\equiv}}\in D^r_{\equiv}\}$ satisfying
$S_{[\mb{t}]_{\equiv}}\in \mathbb{S}_{[\mb{t}]_{\equiv}}$ for every ${[\mb{t}]_{\equiv}}\in D^r_{\equiv}$. 
We associate to $S$ the constraint $( \mb{s},U )$ where the scope $\mb{s}$ is constructed as follows. Before defining $\mb{s}$ we need some preparation. Recall that every coordinate of $U$, and hence of $\mb{s}$, is associated to a tuple $\mb{t}\in D^r$, so we can talk of the class $[\mb{t}]_{\equiv}$ to which each coordinate belongs. In particular, there are $k_{[\mb{t}]_{\equiv}}$ coordinates in $\mb{s}$ of class $[\mb{t}]_{\equiv}$. Hence, by fixing some arbitrary ordering we can use $\mb{s}_{[\mb{t}]_{\equiv}}^i$, $i=1,\dots,k_{[\mb{t}]_{\equiv}}$ to refer to the coordinates in $\mb{s}$ of class $[\mb{t}]_{\equiv}$. Then, informally, $S_{[\mb{t}]_{\equiv}}$ describes which variables 
from $v^0_{[\mb{t}]_{\equiv}},\dots,v^{k_{[\mb{t}]_{\equiv}}n-1}_{[\mb{t}]_{\equiv}}$ to use in order to fill coordinates $\mb{s}_{[\mb{t}]_{\equiv}}^i$, $i=1,\dots,k_{[\mb{t}]_{\equiv}}$. Formally, 
for every $[\mb{t}]_{\equiv}\in D^r_{\equiv}$ and each $i=1,\dots,k_{[\mb{t}]_{\equiv}}$, $\mb{s}_{[\mb{t}]_{\equiv}}^i$ is assigned to the $i^{th}$ element in $S_{[\mb{t}]_{\equiv}}$ in increasing order.

We add such a constraint for each of the $\Pi_{[\mb{t}]_{\equiv} \in D^{r}_{\equiv}} n^{k_{[\mb{t}]_{\equiv}}} =  n^{(|D|^{r})}$ possible choices for $S$. Therefore, after the first stage we have exactly $n^{(|D|^{r})}$ constraints. 

In the second stage we add more constraints which will yield the particular symmetry of $I_{1}$. Note that every permutation $\sigma$ on $\{1,\dots,r\}$ induces a permutation $\sigma'$ on the coordinates of $U$ in a natural way. Specifically, if coordinate $i$ of $U$ is associated to tuple $\mb{t}_i$, then $\sigma'(i)=j$
where $\mb{t}_j=(\mb{t}_i[\sigma(1)],\dots,\mb{t}_i[\sigma(r)])$. Then, in the second stage, for each permutation $\sigma$ on $\{1,\ldots,r\}$ and for every constraint $( \mb{s},U )$ added in the first stage we add the constraint $( \mb{s}',U )$ where for every $1\leq i\leq |D|^r$, 
$\mb{s}'[i]=\mb{s}[\sigma'(i)]$. Therefore, after the second stage we have a total of $m=r!\cdot n^{(|D|^{r})}$ constraints as needed.

We now turn to $I_{2}$. The constraints are constructed in a similar way, but instead of using the family 
$\mathbb{S}$ in the first stage, we use a different family $\mathbb{S'}$.
In particular, for each class $[\mb{t}]_{\equiv}$, $\mathbb{S'}_{[\mb{t}]_{\equiv}}$ is obtained by partitioning $V_{[\mb{t}]_{\equiv}}$ in $k_{[\mb{t}]_{\equiv}}$ blocks of consecutive elements, so that each block has exactly $n$ elements. Then, $\mathbb{S'}_{[\mb{t}]_{\equiv}}$ contains the $n^{k_{[\mb{t}]_{\equiv}}}$ sets that can be obtained by selecting one element from each block. The second stage is done exactly as in $I_{1}$.

\begin{claim} \label{claimone}
$I_1$ and $I_2$ have the same iterated degree sequence.
\end{claim}

\begin{claimproof} Let $[\mb{t}]_{\equiv} \in D^{r}_{\equiv}$. First, we observe that in both instances after the first stage, every variable of $V_{[\mb{t}]_{\equiv}}$ appears in the same number of constraints. More specifically, every variable in $V_{[\mb{t}]_{\equiv}}$ appears in an $n$-fraction of the constraints added in stage $1$. In the case of instance $I_1$ this is due to the fact that $\mathbb{S}_{[\mb{t}]_{\equiv}}$ satisfies condition (b) in Lemma \ref{le:technical} and in instance $I_2$ this follows from the fact that $\mathbb{S}'_{[\mb{t}]_{\equiv}}$ contains all possible sets obtained by choosing an element within each one of the blocks of size $n$.
After the second stage (in both $I_1$ and $I_2$ since the second stage is common) every variable in $V_{[\mb{t}]_{\equiv}}$ still participates in an $n$-fraction of the total number of constraints. In addition, it follows easily that the positions of the scope in which a variable in $V_{[\mb{t}]_{\equiv}}$ participates distribute evenly among the $k_{[\mb{t}]_{\equiv}}$ positions associated to ${\mb{t}}$. That is, in both instances, we have that for every $[\mb{t}]_{\equiv}\in D^r_{\equiv}$, every variable $x\in V_{[\mb{t}]_{\equiv}}$, and every position $i$ associated to $[\mb{t}]_{\equiv}$ there are exactly $\frac{m}{n k_{[\mb{t}]_{\equiv}}}$ constraints in which $x$ appears at position $i$ of the scope, where $m = r! \cdot n^{|D|^r}$.
Using this fact it is very easy to prove that $I_1$ and $I_2$ have the same iterated degree sequence. Formally, one could show by induction on $k$ that for every $[\mb{t}]_{\equiv}\in D^r_{\equiv}$ and $x_1,x_2\in V_{[\mb{t}]_{\equiv}}$, $\delta_k^{I_1}(x_1)=\delta_k^{I_2}(x_2)$ and that for any two constraints $c_1,c_2$ in $I_1$ and $I_2$ respectively $\delta_k^{I_1}(c_1)=\delta_k^{I_2}(c_2)$. Here we are using $\delta_k^{I_1}(\cdot)$ and $\delta_k^{I_2}(\cdot)$ to denote the $k^{th}$ degree of a node in the factor graphs of $I_1$ and $I_2$ respectively.\end{claimproof}

\begin{claim} \label{claimtwo}
Instance $I_{1}$ is unsatisfiable while instance $I_{2}$ is satisfiable.\end{claim}

\begin{claimproof} We start by showing that $I_{1}$ is not satisfiable. Assume by contradiction that $I_{1}$ has a satisfying assignment $\nu$. For each class $[\mb{t}]_{\equiv}$, consider the values given by $\nu$ to the first $d$ variables $v_0,\dots,v_{d-1}$ in $V_{[\mb{t}]_{\equiv}}$. Since $d=k_{[\mb{t}]_{\equiv}} \cdot |D|$, it follows by the pigeon-hole principle that at least $k_{[\mb{t}]_{\equiv}}$ of these variables are assigned by $\nu$ to the same value of $D$. Let $S_{[\mb{t}]_{\equiv}}$ be a subset of $V_{[\mb{t}]_{\equiv}}$ containing $k_{[\mb{t}]_{\equiv}}$ of these variables (we know that this subset belongs to $\mathbb{S}_{[\mb{t}]_{\equiv}}$ by condition (a) of Lemma \ref{le:technical}). Now consider the constraint $( \mb{s},U )$ in $I_{1}$ associated to $S:=\{S_{[\mb{t}]_{\equiv}} \mid {[\mb{t}]_{\equiv}\in D^r_{\equiv}}\}$, which belongs to $I_{1}$. If $\nu$ is a solution to $I_{1}$, then the restriction of $\nu$ to $\mb{s}$ corresponds to an $r$-ary polymorphism of $\Gamma$. But $\nu$ assigns the same value to any two related tuples $\mb{t} \equiv \mb{t}'$, which implies that $\nu$ is symmetric, a contradiction.

We now turn our focus to $I_{2}$. Let $f$ be any $r$-ary polymorphism of $\Gamma$ (for example the $i^{th}$ $(1\leq i\leq r)$ projection operation defined as $f(x_1,\dots,x_r)=x_i$). We shall construct a solution $\nu$ of $I_{2}$ in the following way. Recall that in the definition
of $I_{2}$ we have partitioned the tuples of $V_{[\mb{t}]_{\equiv}}$ in $k_{[\mb{t}]_{\equiv}}$ consecutive blocks. In the first stage, all the elements in each block are placed in the same coordinate of $U$. So, if $\mb{t}_1,\dots,\mb{t}_{|D|^{r}}$ are the tuples associated to coordinates $1,\dots,|D|^{r}$ and hence block $1,\dots,|D|^{r}$ respectively, then we only need that all variables in the $i^{th}$ block are assigned to $f(\mb{t}_i)$ to satisfy all constraints added in the first stage. This assignment also satisfies the constraints added in the second stage, because if $f$ is an $r$-ary polymorphism of $\Gamma$,
then for every permutation $\sigma$ on $\{1,\dots,r\}$, the operation $g(x_1,\dots,x_r)$ defined as $f(x_{\sigma(1)},\dots,x_{\sigma(r)})$ is also a polymorphism of $\gm$. 
\end{claimproof}

To sum up, we constructed two instances $I_{1}$ and $I_{2}$, the latter of which is satisfiable while the former is not, which have the same iterated degree sequence. It follows from Corollary \ref{cor:sameDegSeq->sameSol} that any distributed algorithm will give the same output on both instances, meaning that no algorithm can solve $\dcsp(\{U\})$. From Remark \ref{re:indicator} then it follows that there are also no algorithms that solve $\dcspgm$.

\end{proof}

\subsection{Tractable Languages} \label{sec:tractable}

In this section we turn our attention to the tractable case. In particular we shall show the following:

\begin{theorem}
\label{th:tractable}
Let $\Gamma$ be a constraint language that is invariant under symmetric polymorphisms of all arities. Then there is an algorithm $\Alg$ that solves $\dcspgm$. The total running time, number of rounds, and maximum message size of $\Alg$ are, respectively, $\mathcal{O}(n^3m\log n)$, $\mathcal{O}(n^2)$, and $\mathcal{O}(m\log n)$ where $n$ and $m$ are the number of variables and constraints, respectively, of the input instance.
\end{theorem}

Note that this implies the ``if'' part of Theorem \ref{thm:maintheorem-dec}. $\Alg$ is composed of two phases. In the first phase, a distributed version of the colour refinement algorithm allows every process to calculate its iterated degree. Then, thanks to Theorem \ref{thm:sympol=iso=idsi} we can use the degree of a variable as its ID for the second phase, implying that a distributed adapted version of the $jpq$-consistency algorithm \cite{kozik2020solving} where messages are tagged with a process' iterated degree solves the decision problem for $\gm$.

\paragraph*{Distributed Colour Refinement.} Let $I=(A,X,D,C,\alpha)$ be an instance of $\dcspgm$ and let $n=|X|$ and $m=|C|$. 
There is a very natural way to calculate an agent's iterated degree in a distributed way, both for variables and for constraints. This is a mere adaptation of the $1$-dimensional Weisfeiler-Leman algorithm, also known as colour refinement, an algorithm that partitions the vertices of a graph into classes by iteratively distinguishing them on the basis of their degree (see for example \cite{grohe2017color, grohe2020word2vec}). The algorithm proceeds in rounds. At round $k=0$, each agent $\alpha(v)$ for $v\in X\cup C$ computes $\delta_0(v)$ and broadcasts it to all its neighbours. At round $k>0$, each agent $\alpha(v)$ knows the $(k-1)^{th}$ degrees of its neighbours which it had received in the previous round, uses them to compute $\delta_k(v)$, and broadcasts it to its neighbours. If $k=2n$ (see Proposition \ref{prop:2nisEnough} in the Appendix)
then for every $x,x'\in X$ satisfying $x\sim^k_{\delta} x'$ we have that $x\sim_{\delta} x'$, which implies that we can essentially regard the $k^{th}$ iterated degree as the unique common ID for all variables that are iterated degree equivalent. Then in $2n$ rounds each agent $\alpha(v)$ can compute $\delta_{\infty}(v)$, where we use $\delta_{\infty}$ 
as a shorthand for $\delta_{2n}$. As we described it, the distributed colour refinement algorithm is not particularly efficient in terms of message complexity. Although this is not necessary to achieve polynomial time, we can reduce the space required to encode $\delta_{\infty}(v)$.
\begin{restatable}{lemma}{smaxsize} \label{le:smaxsize}
 Let $s_{\max}$ denote the size of the encoding of $\delta_{\infty}(v)$. A modified version of the distributed colour refinement algorithm that runs over $\mathcal{O}(n^{2})$ rounds achieves $s_{\max} = \mathcal{O}(\log n)$. The time at each round and the maximum size of a message are both bounded above by $\mathcal{O}(m s_{max})$.
\end{restatable} 

As we will see, the price of an increase in the number of rounds (from $n$ to $n^2$) is compensated by the effect of $s_{\max}$ on both time complexity and the size of the messages. 

\paragraph*{The Distributed Consistency Algorithm.} It is well known that if a constraint language $\Gamma$ has symmetric operations of all arities then it satisfies the so-called \emph{bounded width} property (see \cite{Barto2017}). We avoid introducing the definition of bounded width as it is not needed in our results and refer the reader to \cite{Barto2017} for reference. Then, it has been shown in \cite{kozik2020solving} that if $\Gamma$ has bounded width and $I\in \csp(\Gamma)$ satisfies a combinatorial condition called $jpq$-consistency, then $I$ has a solution. Instead of stating literally the result in \cite{kozik2020solving} we shall state a weaker version that uses a different notion of consistency, more suitable to the model of distributed computation introduced in the paper. 

A \emph{set system} $S$ is a subset of $X\times D$. We shall use $S_{x}$ to denote the set $\{d\in D \mid (x,d)\in S\}$. A walk of length $\ell$ (in instance $I$) is any sequence $x_0c_0\dots c_{\ell-1}x_\ell$ where $x_0,\dots,x_\ell$ are variables, $c_0,\dots,c_{\ell-1}$ are constraints, and $x_i,x_{i+1}\in c_i$ for every $0\leq i<\ell$. Note that walks are precisely the walks in the factor graph $G_I$ (in the standard graph-theoretic sense) starting and finishing in $X$. 
 
Let $S$ be a set system, $p$ be a walk, and $B\subseteq S_x$ where $x$ is the starting node of $p$. The \emph{propagation} of $B$ via $p$ under $S$, denoted $B+_{S} p$, is the subset of $D$ defined inductively on the length $\ell$ of $p$ as follows. If $\ell=0$ then $B+_{S} p=B$. Otherwise, $p=p'c_{\ell-1}x_{\ell}$ where
$p'$ is a path of length $\ell-1$ ending at $x_{\ell-1}$. Let $c_{\ell-1}=( \mb{s},R )$. Then we define
$B+_{S} p$ to contain all $e\in D$ such that there exists $d\in B +_{S} p'$ and $\mb{t}\in R$ such that for every $1\leq i\leq arity(R)$, $\mb{t}[i]$ satisfies the following conditions: 
\begin{enumerate}
    \item $\mb{t}[i]\in S_{\mb{s}[i]}$,
    \item if $\mb{s}[i]=x_{\ell-1}$ then $\mb{t}[i]=d$, and
    \item if $\mb{s}[i]=x_{\ell}$ then $\mb{t}[i]=e$.
\end{enumerate}

We are now ready to state the result from \cite{kozik2020solving} that we shall use.

\begin{theorem}[follows from \cite{kozik2020solving}]
\label{the:kozik}
 Let $I$ be an instance of $\csp(\gm)$ where $\gm$ has bounded width and let $S$ be a set system such that $S_{x} \neq\emptyset$ for every $x\in X$ and such that for every walk $p$ starting and finishing at the same node $x$ and for every $d \in S_{x} $, $d$ belongs to $\{d\} +_{S} p$. Then $I$ is satisfiable.
 \end{theorem}
 
Our goal is to design a distributed algorithm that either correctly determines that an instance $I$ is unsatisfiable, or produces a set system $S$ verifying the conditions of Theorem \ref{the:kozik}. This is not possible in general due to the fact that agents are anonymous and hence a hypothetical algorithm that would generate a walk in a distributed way would be unable to determine if the initial and end nodes are the same. However, thanks to the structure established by Theorem \ref{thm:sympol=iso=idsi}, this difficulty can be overcome when $\Gamma$ has symmetric polymorphisms of all arities because, essentially, the iterated degree of a node can act as its unique identifier. To make this intuition precise we will need to introduce a few more definitions.

We say that a pair $(x,d)\in S$ is $S$-\emph{supported} if for every walk $p$ starting at $x$ and finishing at a node $y$ with $x \sim_{\delta} y$, we have that $\{d\}+_S p$ contains $d$. 

\begin{remark}\label{re:boundwalk}
We note that if $(x,d)\in S$ is not $S$-supported and $p=x_0c_0\dots x_{\ell}$ is a walk of minimal length among all walks witnessing that $(x,d)$ is not $S$-supported then $\ell\leq n2^{|D|}$. Indeed if we let $B_i=\{d\}+x_0c_0\dots x_i$, $i=0,\dots,\ell$ then we have that $(x_i,B_i)\neq (x_j,B_j)$ for every $0\leq i<j\leq \ell$, since otherwise the shorter walk $x_0c_0,\dots,x_i,c_{j},\dots,x_\ell$ would contradict the minimality of $p$. Since there are $n$ choices for each $x_{i}$ and $2^{|D|}$ choices for $B_{i}$, the bound follows.\lipicsEnd
\end{remark}

We say that a set system $S$ is \textit{safe} if for every solution $\nu\in I$ we have \[\nu(x)=\nu(y) \textnormal{ for all }x,y \in X \textnormal{ with } x \sim_{\delta} y \implies \nu(x)\in S_{x} \textnormal{ for all } x \in X.\]
Then, we have
\begin{restatable}{lemma}{Asafe} \label{le:Asafe} Let $S$ be a safe set system and let $(x,d)\in S$ be a pair that is not $S$-supported. Then $S \setminus\{(x,d)\}$ is safe.
\end{restatable}

Our distributed consistency algorithm (that is, the second phase of $\Alg$) works as follows. Every variable agent $\alpha(x)$ maintains a set $S_{x} \subseteq D$ in such a way that the set system $S$ is guaranteed to be safe at all times. As a result of an iterative process $S$ is modified. We shall use $S^{i}$ to denote the content of $S$ at the $i^{th}$ iteration, where an iteration is, in turn, a loop of $T=2n(2^{|D|}+1)=\mathcal{O}(n)$ consecutive rounds. The rationale behind this exact value will be made clear later. Initially, $S_{x}^0$ is set to $D$ for every $x\in X$. At iteration $i$ for $i \geq 1$, $S^{i}$ is obtained by removing all the elements in $S^{i-1}$ that are not $S^{i-1}$-supported. Then, in at most $n|D|=\mathcal{O}(n)$ iterations we shall obtain a fixed point $S^{\infty}$. 

The key observation is that when $\gm$ has symmetric polymorphisms of all arities, the satisfiability of $I$ can be determined from $S^{\infty}$. Indeed, if $S^{\infty}_x=\emptyset$ for some $x\in X$ then we can conclude from the fact that $S^{\infty}$ is safe and Theorem \ref{thm:sympol=iso=idsi} that $I$ has no solution. Otherwise, $S^{\infty}$ satisfies the conditions of Theorem \ref{the:kozik} and, hence, $I$ is satisfiable.

It remains to see how to compute $S^{i+1}$ from $S^i$. In an initial preparation step for every iteration, every variable agent $\alpha(x)$ sends $S^i_x$ to all its neighbours. To compute $S^{i+1}$ the algorithm proceeds in rounds. All the messages sent are sets containing triplets of the form $(\delta_{\infty},d,B)$ where $d\in D$, $B\subseteq D$, and $\delta_{\infty}$ is the iterated degree of some variable $x\in X$. It follows from the fact that there are at most $n$ possibilities for the degree of a variable that the size of each message is $\mathcal{O}(ns_{\max})$.

The agents controlling variables and constraints alternate. That is, variables perform internal and send events at even rounds and receive messages at odd rounds, while constraints perform internal and send events at odd rounds and receive messages at even rounds. More specifically, in round $j=0$ of iteration $i$, every variable agent $\alpha(x)$ sends to its neighbours the message $M$ containing all triplets of the form $(\delta_{\infty}(x),d,\{d\})$ with $d\in S^i_x$. At round $2j$ for $j>0$, $\alpha(x)$ computes $M=M_1\cup\cdots\cup M_r$ where $M_1,\dots,M_r$ are the messages it received at the end of round $2j-1$. Subsequently, for every triplet $(\delta_{\infty},d,B)\in M$ with $\delta_{\infty}=\delta_{\infty}(x)$ and $d\not\in B$, $\alpha(x)$ marks $d$ as `not $S^{i}$-supported'. Finally, it sends message $M$ to all its neighbours. This computation can be done in time $\mathcal{O}(rns_{\max})=\mathcal{O}(mns_{\max})$ provided that each message is stored as an ordered array.

In round $2j+1$, every constraint agent $\alpha(c)$ computes from the messages $M_x$ (received from each neighbour $\alpha(x)$ in the previous round) the set $M'_x$, which contains for every variable $y \in c$ and every $(\delta_{\infty},d,B)$ in $M_y$, the triplet $(\delta_{\infty},d,B+_{S^i} p)$ where $p=y,c,x$. Finally, it sends to each neighbour $\alpha(x)$ the corresponding message $M'_{x}$. Note that while $\alpha(c)$ doesn't know the address of $\alpha(x)$ specifically, knowing the label of the channel that connects them is sufficient to calculate $M'_{x}$ correctly and send the message accordingly. Moreover, for given $y$ and $x$, $\alpha(c)$ can compute $B+_{S^i} p$ in $\mathcal{O}(1)$ time as $\alpha(c)$ knows both $S^{i}_{y}$ and $S^{i}_x$. Hence, since the arity of the relations is fixed (as $\Gamma$ is fixed) the total running time at iteration $2j+1$ of a constraint agent $\alpha(c)$ is $\mathcal{O}(ns_{\max})$.

Now it is immediate to show by induction that for every $j\geq 0$, every $x\in X$ and $c\in C$ with $x\in c$ the message sent by $\alpha(x)$ to $\alpha(c)$ at the end of round $2j$ is precisely 
\[\{(\delta_{\infty}(y),d,\{d\}+p) \mid y\in X, d\in S^i_y, p \textnormal{ is a walk of length $j$ of the form $p=y,\dots,x$}\}\]
and the message sent by $\alpha(c)$ to $\alpha(x)$ at the end of round $2j+1$ is precisely 
\[\{(\delta_{\infty}(y),d,\{d\}+p)  \mid y\in X, d\in S^i_y, p \textnormal{ is a walk of length $j+1$ of the form $p=y,\dots,c,x$}\}.\]

By Remark \ref{re:boundwalk} only $2n2^{|D|}=T-2n$ iterations are needed to identify all elements in $S^{i}$ that are not $S^{i}$-supported. Hence, after exactly $T-2n$ rounds every variable agent $\alpha(x)$ computes $S^{i+1}_x$ by removing all the elements in $S^{i}$ that are marked as ``not $S^{i}$-supported''. 
If $S_{x}^{i+1} = \emptyset$, then $\alpha(x)$ initiates a wave, which is propagated by all its neighbours, broadcasting that an inconsistency was detected. In this case, in at most $2n$ additional rounds all agents can correctly declare that $I$ is unsatisfiable. Otherwise, a new iteration begins.

To sum up, the distributed consistency algorithm consists of $\mathcal{O}(n)$ iterations consisting, each, of $\mathcal{O}(n)$ rounds where the total running time for internal events at a given round is $\mathcal{O}(mns_{\max})$ and the maximum size of each message transmitted is $\mathcal{O}(ns_{\max})$. Together with the bounds given by Lemma \ref{le:smaxsize} for the distributed colour refinement phase, this completes the proof of Theorem \ref{th:tractable}.

\subsection{The Search Algorithm} \label{sec:search}

We conclude by presenting the proof of Theorem \ref{thm:maintheorem-ser}. The hardness part follows immediately from Theorem \ref{thm:maintheorem-dec} as the search problem is as difficult as the decision problem. For the positive result we shall present an adaptation of the algorithm solving the decision version. 
Let $I$ be an instance of $\dser$ where $\Gamma$ contains symmetric polymorphisms of all arities. In what follows we shall use intensively the fact that $\pol(\Gamma)$ is closed under composition. Let $J\subseteq D$ be minimal with the property that $f(D)=J$ for some unary polymorphism $f$ in $\pol(\Gamma)$. It is fairly standard to show that for every $r\geq 0$ there is a $r$-ary symmetric operation $g$ such that $g(x,\dots,x)=x$ for every $x\in J$. Indeed, let $f$ satisfy $f(D)=J$ and let $g$ be any $r$-ary symmetric polymorphism in $\pol(\Gamma)$. Then the unary operation $h$ defined by
$h(x)=f\circ g(x,\dots,x)$ is a unary polymorphism of $\Gamma$.
By the choice of $f$ we have $h(D)\subseteq J$. We note that $h(J)=J$ since otherwise $h
^2$ would contradict the minimality of $f$. Consequently, $h^{-1}$ belongs to $\pol(\Gamma)$ and, hence,
the $r$-ary operation defined as $h^{-1}\circ f \circ g$ satisfies the claim. This implies that if we enlarge the constraint language by adding all singletons $\{d\}$, $d\in J$, the resulting constraint language, which we shall denote by $\Gamma'$, still has symmetric polymorphisms of all arities. For convenience we also include $D$ in $\Gamma'$.   

The algorithm has two phases. In the first phase it runs the decision algorithm to determine whether the instance is satisfiable. As a byproduct, every variable agent $\alpha(x)$ has computed its iterated degree $\delta_{\infty}(x)$ and knows as well its rank in a prescribed ordering of all variable degrees $\delta_{\infty}^1,\dots,\delta_{\infty}^r$, $r\leq n$. This (partial) order will be used to coordinate between the agents. An $i$-agent, $1\leq i\leq r$ is any agent $\alpha(x)$ with $\delta_{\infty}(x)=\delta^i_{\infty}$. We also assume a fixed ordering on the elements in $D$. If the instance is unsatisfiable nothing else remains to be done so from now on we shall assume that the instance is satisfiable.

In the second phase the algorithm searches for a solution. Every variable agent $\alpha(x)$ maintains a set $F_x\subseteq D$ with the property that there is a solution $\nu$ that \emph{falls within} $F$, i.e, such that $\nu(x)\in F_x$ for every $x\in X$. Initially every agent $\alpha(x)$ sets $F_x=D$ so it is only necessary to make sure that this condition is preserved during the execution of the algorithm. The second phase contains two nested loops. The outer loop has $r$ iterations and the inner loop consists of at most $|D|$ iterations, so we shall use iteration $(i,d)$ to indicate the run of the algorithm at iteration $i=1,\dots,r$ of the outer loop and at iteration $d$ of the inner loop.  

At the beginning of iteration $(i,d)$ every variable agent $\alpha(x)$ defines $S_x \subseteq D$ to be $S_x=\{d\}$ whenever $\alpha(x)$ is an $i$-agent and $S_x=F_x$ elsewhere. Then it runs the distributed consistency algorithm starting at $S$ obtaining a fixed point $S^{\infty}$. We note that since all initial sets $S_x$ belong to $\Gamma'$ and $\Gamma'$ contains symmetric polymorphisms of all arities then the obtained fixed point $S^{\infty}$ correctly determines whether there exists a solution $\nu$ that falls within $S$. Then every $i$-agent $\alpha(x)$ checks whether $S^{\infty}_x=\emptyset$. In case of positive answer nothing else is done and round $(i,d)$ finishes. Otherwise, $\alpha(x)$ sets $F_x$ to $\{d\}$ and starts a wave to indicate to all processes that the $i^{th}$ iteration of the outer loop is finished and that the next iteration of the outer loop can start. 
When the $r$ iterations of the outer loop have been completed the set system $F$ contains only singletons. The assignment that sets every variable $x\in X$ to the only element in $F_x$ is necessarily a solution. This concludes the proof of Theorem \ref{thm:maintheorem-ser}.

\section{Conclusion}

We analysed the complexity of the distributed constraint satisfaction problem on a synchronous, anonymous network parametrised by the constraint language. We showed that, depending on the polymorphisms of $\gm$, $\dcspgm$ is either solvable in polynomial time, or not solvable altogether. A number of natural open questions arise in this context. For instance, it is not clear whether asynchronous networks are strictly more powerful than their synchronous counterpart. Moreover, it would be interesting to explore the role of allowing agents to make random choices - provided this is not used to create and share unique IDs.

In the spirit of \cite{grohe2007complexity}, one could consider characterizing the structural restrictions on tractable distributed CSPs, or in other words, determining which \textit{classes of networks} are tractable in the DCSP framework, regardless of the constraint language. The starting point for this analysis could be the work on fibrations by Boldi et al. (see for example \cite{Boldi1996Symmetry, Boldi2001effective}). In particular, we propose the question of establishing a connection between the universal fibration of a graph and its iterated degree sequence.

\bibliography{ms}

\appendix

\section{Proofs from Section \ref{sec:prelims}}

\begin{remark} \label{re:alternativemodel}
Throughout the paper, we assumed that both variables and constraints are controlled by agents in a distributed network (throughout this section, we will refer to this as model 1). However, when all the constraint relations in $\Gamma$ are binary it is also valid and, indeed, more common to assume that only variables are controlled by agents, and there is a communication channel between any two variable agents $\alpha(x)$ and $\alpha(x')$ whenever $x$ and $x'$ share a constraint (model 2) which is labelled with the constraint relation and the direction of the constraint.

It is very easy to see that in the binary case both models are equivalent. Indeed, for every CSP instance $(X,D,C)$, let $(A_1,X,D,C,\alpha_1)$ and $(A_2,X,D,C,\alpha_2)$ be the associated DCSP instances in model $1$ and $2$ respectively. It is easy to see that every algorithm in model 2 can be easily simulated by an algorithm in model 1. In particular, it is only necessary that at round $2j$ every variable agent $\alpha_1(x)$ replicates the $j^{th}$ round of $\alpha_2(x)$ (while every constraint agent $\alpha_1(c)$ remains idle). Then, round $2j+1$ is used to replicate the messages sent at round $j$. That is, whenever $\alpha_2(x)$ sends a message to a neighbour $\alpha_2(x')$ at round $j$, $\alpha_1(x)$ sends a message to $\alpha_1(c)$ at round $2j$, where $c$ is the constraint shared by $x$ and $x'$. At round $2j+1$ then $\alpha_1(c)$ forwards the message to $\alpha_1(x')$.

Similarly, any algorithm in model 1 can be replicated in model 2. In this case, at a given round $j$, every agent $\alpha_2(x)$ simulates the internal computation done at round $j$ by $\alpha_1(x)$ and all its neighbours.\lipicsEnd
\end{remark}

\begin{proposition} \label{prop:2nisEnough} 
Let $I=(A,X,D,C,\alpha)$ be an instance of $\dcspgm$ and let $v,v' \in X \cup C$. Let $k \geq 2n$ where $n = |X|$. Then, $v \sim_{\delta}^{k} v'$ implies $v \sim_{\delta} v'$.
\end{proposition}

\begin{proof}
We start by showing that for all non-negative integers $k,k'$ with $k \leq k'$, the partition induced by $\sim_{\delta}^{k'}$ on $X \cup C$ is at least as refined as the partition induced by $\sim_{\delta}^{k}$.
The proof goes by induction. Let $v,v' \in G_{I}$. Clearly if $\delta_{0}(v) \neq \delta_{0}(v')$, then $\delta_{k}(v) \neq \delta_{k}(v')$ for all $k \in \mathbb{N}$, so in particular $\delta_{1}(v) \neq \delta_{1}(v')$. Now assume that $\delta_{k}(v) = \delta_{k}(v')$ implies $\delta_{k-1}(v) = \delta_{k-1}(v')$. Then it is a clear consequence of the definition of $\delta_{k}$ that
$\delta_{k+1}(v) = \delta_{k+1}(v') \implies \delta_{k}(v) =  \delta_{k}(v')$ too as required.

Now it remains to show that if $v \sim_{\delta}^{2n} v'$, then $v \sim_{\delta}^{k} v'$ for all $k \geq 2n$. The result is immediate if we replace $2n$ by $n+m$. To achieve $2n$ we use the fact that the factor graph is bipartite. Denote by $P^{k}$ and $Q^{k}$ the partitions induced by $\sim_{\delta}^{k}$ on $X$ and $C$ respectively and note that if $P^{k-2}=P^{k}$ then $(P^{k},Q^{k})$ is a fixed point. We notice that $P^{k-2}=P^{k}$ must occur for some $k \leq 2n$ and we are done. 
\end{proof}

\idsisameSol*
\begin{proof} $(\Rightarrow)$. At the beginning of the algorithm, all processes are in the same state. Let $v$ be a node in the factor graph of $I$, and denote by $m_{t}(v)$ the message broadcast at time $t$ by $\alpha(v)$ to its neighbours. For any two nodes $v$, $v'$, $\delta_{1}(v) = \delta_{1}(v')$ is equivalent to $v$ and $v'$ having the same knowledge at the start of the algorithm. This means that the first internal and send events are the same at $\alpha(v)$ and at $\alpha(v')$, hence $m_{1}(v)=m_{1}(v')$. Then, it is easy to see by induction that $\delta_{t}(v)= \delta_{t}(v') \implies m_{t}(v)=m_{t}(v')$, which in turn implies that \[v \sim_{\delta} v' \implies m_{t}(v) = m_{t}(v') \quad \textnormal{at all times } t = 1,2,\ldots\]
This implies that at any time $t$, $\alpha(v)$ and $\alpha(v')$ send and receive the same messages, so they have the same knowledge and hence the internal events at $\alpha(v)$ and $\alpha(v')$ are the same at all time. In particular, if the algorithm terminates, then the terminating state is the same at $\alpha(v)$ and $\alpha(v')$, and therefore the decision and, in case of search, the value of $\nu$ at $\alpha(v)$ and $\alpha(v')$ are the same.

$(\Leftarrow)$. Consider the algorithm that calculates the iterated degree at each node (we detail the procedure in the proof of Theorem \ref{th:tractable}). If $v \not \sim_{\delta} v'$, then we can find an algorithm that on the basis of the iterated degree gives different outputs at $\alpha(v)$ and $\alpha(v')$.
\end{proof}

\section{Proofs from Section \ref{sec:BLP}}

\thmrounding*
\begin{proof}
It is clear that if there is no feasible solution to $\BLP(I)$, then $I$ is not satisfiable, regardless of the polymorphisms of $\gm$. So let $\gm$ be such that $\pol(\gm)$ contains symmetric polymorphisms of all arities and let $I$ be an instance of $\cspgm$ such that $\BLP(I)$ is feasible. We can assume that there exists some natural number $n$ such that in the feasible solution to $\BLP(I)$, all variables take rational values of the form $\frac{r}{n}$ for some integer $r$. That is, for all $x \in X$, $d \in D$, $c \in C$ and $\mb{t} \in R_{c}$ (where $R_{c}$ is the relation of $c$) there exist corresponding integers such that \[v(x,d)=\frac{r(x,d)}{n} \quad \textnormal{ and } \quad v(c,\mb{t})=\frac{r(c,\mb{t})}{n}.\]

Let $f$ be a symmetric polymorphism of $\Gamma$ of arity $n$. For every $x \in X$ we shall denote by $f_{x}$ the value of $f$ when applied to an $n$-tuple where each $d \in D$ appears exactly $r(x,d)$ times. We claim that the assignment $\nu: X \to D$ given by $\nu(x) = f_{x}$ satisfies $I$. To see this, consider an arbitrary constraint $c = (\mb{s}, R)$. Denote by $\mb{t}'$ the tuple obtained by applying $f$ coordinate-wise to $n$ tuples $\mb{t}_{1}, \ldots, \mb{t}_{n}$ chosen as follows: each tuple $\mb{t} \in R$ is chosen exactly $r(c,\mb{t})$ times. Clearly $\mb{t}' \in R$ since $f$ is a polymorphism of $R$. So, to show that $c$ is satisfied by $\nu$, it is enough to show that $\nu(\mb{s}[i]) = \mb{t}'[i]$. Now, $\mb{t}'[i]$ is the result of applying $f$ to the set of the $i^{th}$ elements of $\mb{t}_{1}, \ldots, \mb{t}_{n}$. But any $d \in D$ occurs in $\mb{t}_{1}[i], \ldots, \mb{t}_{n}[i]$ exactly $\sum_{ \mb{t}[i]=d} r(c,\mb{t}) = n \cdot  v(\mb{s}[i],d) = r(\mb{s}[i],d)$ times, and so, given that $f$ is symmetric, we have \[\mb{t}'[i]=f_{\mb{s}[i]} = \nu(\mb{s}[i])\] as required.
\end{proof}

\thLP*
\begin{proof}
We start by rewriting the program in the form
\begin{equation} \tag{\ref{eq:feasibProblem}}
    \exists? \mb{v} \in [0,1]^{V} \quad B\mb{v} \geq \mb{b}.
\end{equation}
by replacing every equality $a=b$ by the inequalities $a \geq b$ and $-a \geq -b$.

It will be convenient to index the rows and columns of $B$ not using positive integers. Let us start with the columns. Each column is associated to a variable of $\BLP(I)$, i.e, a variable of the form $v(x,d)$, $x\in X,d\in D$ or 
$v(c,\mb{t}), c\in C, \mb{t}\in R_c$. In the first case, we index the corresponding column with the pair $(x,d)$ whereas in the second case we index it with the pair $(c,\mb{t})$, and we denote by $V$ the set of all such indices. 

Now, let us turn our attention to the rows. Every equation in (\ref{eq:LP2}) gives rise to two rows that we shall index with $(x,+)$ and $(x,-)$. Similarly, every equation in (\ref{eq:LP3}) also gives rise to two rows that we shall index with $(c,i,d,+)$ and $(c,i,d,-)$. Let us denote by $W$ the set of all indexes for rows. 

We shall see later how to define an oracle which, given a probability $W$-vector $\mb{p}$
{(i.e, a vector $\mb{p}$ with non-negative entries such that the sum of all its entries is $1$)}, outputs a vector $\mb{v}$ which is a solution to the weaker problem
\begin{equation}
\tag{\ref{eq:probFeasibProblem}}
    \exists? \mb{v} \in [0,1]^{V} \quad \mb{p}^{T}B\mb{v} \geq \mb{p}^{T}\mb{b}
\end{equation}
if one exists, or correctly states that no such vectors exist otherwise. Note that if a solution exists to ($\ref{eq:feasibProblem}$), then it is necessarily also a solution to ($\ref{eq:probFeasibProblem}$), while the opposite is not true in general.

For every $w\in W$, let us denote by $B_w$ the row corresponding to $w$. If $w=(x,+)$ then, since the vector returned by the oracle satisfies 
$\mb{v}\in[0,1]^V$ it follows easily that $B_w\mb-\mb{b}[w]\in [-1,|D|]$. Similarly, if $w=(c,i,d,+)$ then $B_w\mb{v}-\mb{b}[w]\in [-1,\max_{R \in \gm}|R|]$. It follows that by setting $\ell=1$ and $\rho=\max\{|D|,\max_{R \in \gm}|R|\}$ any such oracle-given vector $\mb{v}$ satisfies the following condition: there is a fixed subset $J \subseteq W$ (consisting precisely of the positive rows) such that
\begin{align*}
   B_{w} \mb{v} -\mb{b}[w] \in [-\ell,\rho] \quad \forall w \in J,\\
   B_{w} \mb{v} -\mb{b}[w] \in [-\rho,\ell] \quad \forall w \not\in J.
\end{align*}
Such an oracle is known as a \textit{$(\ell,\rho)$-bounded oracle}. Then we have:

\begin{theorem}[\cite{Arora2012}] \label{thm:oracle->MWU} 
Let $\eps > 0$ be an arbitrary error parameter. Suppose that there exists an $(\ell,\rho)$-bounded oracle for the feasibility problem \textnormal{(\ref{eq:probFeasibProblem})}. Assume that $\ell \geq \frac{\eps}{2}$. Then there exists an algorithm which either finds $\mb{v}$ such that 
$B \mb{v} \geq \mb{b} - \eps$ whenever such $\mb{v}$ exists, or correctly concludes that no such $\mb{v}$ exists otherwise. Such algorithm makes $\mathcal{O}(\ell \rho \log(|W|) / \eps^{2})$ calls to the oracle.
\end{theorem}

The algorithm that Theorem \ref{thm:oracle->MWU} refers to is Multiplicative Weight Update (MWU), a well-known technique that is widely used in optimisation and machine learning. MWU was discovered independently by researchers of different communities; for a survey of its different variants we refer the reader to [2]. The version that is relevant to our paper is described in Algorithm \ref{alg:MWU}.
Recall that the algorithm assumes that there is a feasible solution.

\setcounter{algocf}{0}
\begin{algorithm}[t]
\SetAlgoLined
\textbf{Initialisation:} Fix $\eta \leq \frac{1}{2}$ and let $\mb{w}^{(1)}$ be a $W$-vector, whose entries, called weights, are initially set to $1$. \\
\For{$t = 1, \ldots, T$ }{
Compute the probability vector $\mb{p}^{(t)} = \frac{1}{\Phi(t)} \mb{w}^{(t)}$, where $\Phi(t) = \sum_{j=1}^{|W|} \mb{w}^{(t)}[j]$\\
Let $\mb{v}^{(t)}$ be a solution satisfying $(\mb{p}^{(t)})^{T}B\mb{v}^{(t)} \geq (\mb{p}^{(t)})^{T}\mb{b}$ given by oracle $\oo$\\
 Compute the losses $\pmb{\ell}^{(t)} = \frac{1}{\rho} (B \mb{v}^{(t)}- \mb{b})$ \\
Compute the new weights $\mb{w}^{(t+1)} = \mb{w}^{(t)}(1-\eta \pmb{\ell}^{(t)}$) }
\KwRet $\mb{v}:= \frac{1}{T} \sum_{t=1}^{T} \mb{v}^{(t)}$\\
 \caption{Multiplicative Weight Update}
\end{algorithm}

We shall see that if we choose the oracle $\oo$ wisely then for every $x, x' \in X$ with $x \sim_{\delta} x'$ and every $d\in D$, the solution returned by the MWU algorithm assigns the same value to $v(x,d)$ and $v(x',d)$. 

To see this we need some more notation. We note that $\sim_{\delta}$ induces in a natural way an equivalence relation $\sim_V$ on $V$. In particular,
we have that $v,v'\in V$ are $\sim_V$-related if $v=(x,d)$ and $v'=(x',d)$ where $x\sim_{\delta} x'$ and $d\in D$, or $v=(c,\mb{t})$ and $v'=(c',\mb{t})$ 
where $c\sim_{\delta} c'$ and $\mb{t} \in R_c$ (note that, necessarily, $R_c=R_{c'}$). Similarly $\sim_{\delta}$ induces an equivalence relation, denoted $\sim_W$, on $W$. More specifically, we have that $w,w'\in W$ are $\sim_W$-related if $w=(x,s)$ and $w'=(x',s)$ where $x\sim_{\delta} x'$ and $s\in\{+,-\}$ or $w=(c,i,d,s)$ and $w'=(c',i,d,s)$ where $c\sim_{\delta} c'$, 
$i\in\{1,\dots,arity(c)\}$, $d\in D$, and $s\in\{+,-\}$.

Now, we say that a $V$-vector $\mb v$ is $\sim_V$-\emph{preserving} if $\mb{v}[v]=\mb{v}[v']$ whenever $v\sim_V v'$ and we similarly define $\sim_W$-preserving $W$-vectors. So it is enough to show that there exists some oracle $\oo$ that guarantees that at each iteration $t$ of the WMU algorithm, $\mb{v}^{(t)}$ is $\sim_V$-preserving. To this end we need the following easy properties.
\begin{claim} For all $\sim_V$-preserving $V$-vectors $\mb{v}$ and for all $\sim_W$-preserving $W$-vectors $\mb{w}$, we have that
\begin{enumerate}
\item $B\mb{v}$ is $\sim_W$-preserving; 
\item ${\mb{w}}^T B$ is $\sim_V$-preserving. 
\end{enumerate}
\end{claim}

\begin{claimproof}
We include only the proof of (2) as the proof of (1) is analogous and, indeed, simpler. Let $\mb{v}:={\mb{w}}^T B$. An easy computation shows that 
\begin{equation*} \label{eq:coeffxd}
    \mb{v}(x,d) =  \mb{w}(x,+)-\mb{w}(x,-)- \sum_{c\in C_x} \sum_{\substack{1\leq i\leq arity(c) \\ \mb{s}_{c}[i]=x}} \big( \mb{w}(c,d,i,+)-\mb{w}(c,d,i,-) \big)
\end{equation*}
where we write $C_x$ to denote the set of all constraints in $C$ where $x$ appears in the scope, and \begin{equation*} \label{eq:coeffcd}
    \mb{v}(c,\mb{t}) =  \sum_{\substack{1 \leq i \leq arity(c) \\ \mb{t}[i]=d}} \big( \mb{w}(c,d,i,+) - \mb{w}(c,d,i,-) \big)
\end{equation*}
It is immediate to see that, if $\mb{w}$ is $\sim_{W}$-preserving, then $\mb{v}(c,\mb{t})=\mb{v}(c',\mb{t})$ whenever $c\sim_{\delta} c'$. Let us show that $\mb{v}(x,d)=\mb{v}(x',d)$ whenever $x\sim_{\delta} x'$. Since $\mb{w}$ is $\sim_W$-preserving we have that $\mb{w}(x,s)=\mb{w}(x',s)$
for $s\in\{+,-\}$ and hence we only need to show that $\varphi_x(C_x)=\varphi_{x'}(C_{x'})$ where $\varphi_x(C_x)$ is a shorthand for
\[\sum_{c\in C_{x}} \sum_{\substack{1\leq i\leq arity(c) \\ \mb{s}_{c}[i]=x}} \big(\mb{w}(c,d,i,+)-\mb{w}(c,d,i,-) \big)\]
and $\varphi_{x'}(C_{x'})$ is defined analogously.

Now, for every $R\in\Gamma$, every $S \subseteq \{1,\dots,arity(R)\}$, and every class $[c]$ of equivalent constraints, let $C_{x,R,S,[c]}$ be the set of constraints in $C_x$ that belong to $[c]$, whose constraint relation is $R$, and whose scope $\mb{s}$ satisfies the following: $(i\in S)\Leftrightarrow \mb{s}[i]=x$ for every $i\in \{1,\dots,arity(R)\}$. Note that since $C_x$ and $C_{x'}$ can be partitioned as the union of sets of this form it is only necessary to show
that $\varphi_x(C_{x,R,S,[c]})=\varphi_{x'}(C_{x',R,S,[c]})$ for every choice of $R$, $S$, and $[c]$. To see this it is enough to note that $|C_{x,R,S,[c]}|=|C_{x',R,S,[c]}|$ (because $x\sim_{\delta} x'$) and that, since $\mb{w}$ is $\sim_W$-preserving, for every
constraint $c'\in[c]$ and every choice of $d$, $i$, and $s$, we have $\mb{w}(c',d,i,s)=\mb{w}(c,d,i,s)$.
\end{claimproof}

Now, consider the oracle $\oo$ that, given a $W$-vector $\mb{p}$, returns the $V$-vector $\mb{v}$ defined as $\mb{v}[v]=1$ if $\mb{p}^T B[v]$ is positive and $\mb{v}[v]=0$ otherwise. Since $\mb{v}$ maximizes $\mb{p}^T B\mb{v}$ under the restriction $\mb{v}\in[0,1]^V$ it follows that 
$\mb{v}$ satisfies (\ref{eq:probFeasibProblem}). Furthermore, it is easy to see that if $\mb{p}$ is $\sim_W$-preserving then $\mb{v}$ is $\sim_V$-preserving.

Now, note that by definition both $\mb{w}^{(1)}$ - which is an all-ones $W$-vector - and $\mb{b}$ are $\sim_W$-preserving. It follows easily by induction that for each $t$, $\mb{v}^{(t)}$ is $\sim_V$-preserving and $\mb{w}^{(t)}$ is $\sim_W$-preserving. Hence, if we call algorithm WMU iteratively with $T\rightarrow \infty$ we obtain in the limit a feasible solution satisfying the conditions of the statement. We note here that, although we have not included explicitly any inequalities requiring that all the variables in $\BLP(I)$ take values in the range $[0,1]$, this is guaranteed by the fact that all the entries of the vector returned by $\oo$ are in the range $[0,1]$. This concludes the proof of Theorem \ref{th:LP}. \end{proof}


\section{Concepts from Section \ref{sec:complexity}}

\paragraph*{pp and efpp definability} We shall start by presenting a notion of definability for relations that is closely associated to polymorphisms. A relation $R \subseteq D^{k}$ is said to be \textit{primitive positive definable}, most commonly shortened to \textit{pp-definable}, from a constraint language $\gm$ over the same domain $D$ if there exists a pair
$((x_1,\dots,x_k),I)$ with $I=(X,D,C)$ an instance of $\csp(\Gamma\cup\{\eq_D\})$ where $\eq_D=\{(d,d) \mid d\in D\}$ is the equality relation, and $x_1,\dots,x_k$ are distinct variables in $X$ such that for every tuple $\mb{t} \in D^k$
\[\mb{t} \in R \iff \text{ there is a solution $\nu$ of $I$ such that } \mb{t}[i]=\nu(x_i) \text{ for all } 1\leq i\leq k\]

A constraint language $\gm'$ is pp-definable from $\gm$ if all the relations in $\gm'$ can be pp-defined from $\gm$. The following complexity reduction between CSP classes is well known. 

\begin{theorem}[see \cite{Barto2015CSP}]
If $\gm$ pp-defines $\gm'$, then $\csp(\gm')$ is log-space reducible to $\cspgm$.
\end{theorem}

However, in the distributed setting, allowing equality introduces a few technical difficulties. Fortunately, this obstacle could be overcome by considering a more restricted notion of pp-definability which, following \cite{Jonssonetal17}, we shall call efpp-definability, where equality is not allowed. More precisely, we shall say that relation $R$ is \textit{equality-free primitive positive definable} (efpp-definable, for short) from $\Gamma$ if it is pp-definable and, in addition, the instance $I$ witnessing the pp-definition belongs to $\csp(\Gamma)$. That is, we are not allowed to use the equality relation in the instance, unless, of course, it belongs already to $\Gamma$.
Then, we have:

\begin{proposition}
 \label{prop:DCSPreduction}
Assume that $\gm'$ is efpp-definable from $\gm$. If $\dcspgm$ is solvable in polynomial time (resp. finite time) then so is $\dcsp(\gm')$.
\end{proposition}
\begin{proof} 
Given an algorithm $\Alg$ that solves $\dcspgm$ we can design a new algorithm $\Alg'$ for $\dcsp(\Gamma')$ that given an instance $I'=(A',X',D,C',\alpha')$ of $\dcsp(\gm')$ simulates the execution of $\Alg$ with the instance $I=(A,X,D,C,\alpha)$ of $\dcspgm$ defined as follows. For every 
    constraint $c=((s_1,\dots,s_k),R)$ in $C'$, consider the pair $((x_1,\dots,x_k),I_R)$ defining $R$ and replace 
    constraint $c$ by the instance $I_c$ (meaning all its variables and constraints) obtained from $I_R$ by renaming the variables so that $x_i=s_i$ for every $1\leq i\leq k$ and the rest of variables in $I_c$ are fresh. Then $\alpha$ is defined such that it agrees with $\alpha'$ over $X'$ and, as usual,
    every variable and constraint in $I$ is controlled by a different agent.
    
    The simulation is as follows. At each round, for every $x\in X'$, $\alpha'(x)$ simulates the execution of $\alpha(x)$ as in $\Alg$, and for every $c\in C'$, $\alpha'(c)$ simulates the execution of all constraints and fresh variables in $I_c$. We note that no new communication channels need to be created as this simulation is done internally by $\alpha'(c)$. The transmission of messages can be also easily simulated for every pair of neighbours $\alpha(x)$ and $\alpha(c)$ in $I$. In fact, if $x\not\in X'$, then both $\alpha(x)$ and $\alpha(c)$ are simulated by the same agent $\alpha'(c)$ in $I'$ (and, hence, no communication is required). Otherwise, if $x\in X'$, $\alpha(x)$ is simulated by $\alpha'(x)$ and $\alpha(c)$ is simulated by some neighbour $\alpha'(c')$ of $\alpha(x)$.    
\end{proof}

\begin{remark}
\label{re:indicator}
We note here that for every $r\geq 1$, the indicator problem of order $r$ of $\Gamma$ constitutes an efpp-definition of the $|D|^r$-ary relation $U$ encoding the set of all polymorphisms of arity $r$. It then follows from Proposition \ref{prop:DCSPreduction} that if $\dcspgm$ is solvable in finite time then so is $\dcsp(\{U\})$.\lipicsEnd
\end{remark}

\paragraph*{Proofs from Section \ref{sec:complexity}}

\letechnical*
\begin{proof}
If $k=1$ we can just define $\mathbb{S}$ to be the set containing all singletons in $\{0,1,\dots,kn-1\}$ so we can assume that $k\geq 2$.
Pick some $n$ that is a multiple of $k$ and consider the subsets of ${\{0,1,\dots,kn-1\}}$. We say that one such set is bad if $S=S+i \mymod{kn}$ for some $i \neq 0$, and good otherwise where the right-hand side of the equation is a shorthand for the set ${\{s+i \mod{kn} \mid s\in S\}}$. The following facts hold.
\begin{claim}
If $n>\frac{2(d-1)}{k}$, then all subsets of $\{0,\dots,d-1\}$ are good. \end{claim}

\begin{claimproof}{5}
Let $S \subseteq \{0,\dots,d-1\}$ and assume that $S$ is bad. Then, there exists $i$ such that $S = S+i \mymod{kn}$. Denote by $s_{m}$, $s_{M}$ the smallest and largest elements of $S$ respectively. Then, $i$ must be such that \[kn -(s_{M} - s_{m}) \leq i \leq s_{M} - s_{m}\]
 which implies that $d - 1 < i \leq d-1$, a contradiction.\end{claimproof}

\begin{claim}
There are at least $n^k$ good sets.\end{claim}
\begin{claimproof}{6}
We say that a bad set is \textit{canonical} if it is not the union of bad sets of smaller size. Observe that in a canonical bad set, the distance between every two consecutive elements is constant. That is, we can write $S=\{s+i \mymod{kn}\}_{i=0}^{|S|}$ for some $s \leq kn$ and $0<i\leq kn$. Now, every bad set of size $k$ is a disjoint union of canonical bad sets, and in particular it is the disjoint union of a canonical bad set $S_{1}$ of size $j$ for some $j \in \{2,\ldots, k\}$, and another bad set $S_{2}$ of size $k-j$. Then, to get a loose upper bound on the number of bad sets we notice that there are at most $k(k-2)n$ choices for $S_{1}$ (since we have $kn$ choices for the first element and at most $k-2$ choices for the number of elements in $S_1$), and at most $\binom{kn-j}{k-j} = O(n^{k-2})$ choices for $S_{2}$, which leaves us with at most $O(n^{k-1})$ bad sets. This implies that there are at least $\binom{kn}{k}-O(n^{k-1})$ good sets, which, since $k\geq 2$, is at least $n^k$ for $n$ large enough.
\end{claimproof}

Therefore, consider the collection of good $k$-element subsets of $\{0,1,\dots,kn-1\}$. We say that two sets $S$, $S'$ are \textit{related} if $S=S'+i \mymod{kn}$ for some $i \neq 0$. Note that, since we are only considering good sets, every class of related sets has exactly $kn$ members and, hence, there are at least $n^k/kn$ many classes. Also it is immediate that every class of related sets satisfies condition (b).

Hence, to construct $\mathbb{S}$ we just need to remove some of the classes of good sets so that we end up having exactly $n^k/kn$ classes, which corresponds to $n^{k}$ sets. We have to keep all the classes containing one of the sets of condition (a), which is always possible if we pick $n$ large enough so that $\binom{d}{k} \leq n^k$.
\end{proof}

\smaxsize*

\begin{proof}
 
 We describe a variation of the distributed colour refinement algorithm that achieves the required bounds. After computing the $k^{th}$ degree and before proceeding to compute the $(k+1)^{th}$ degree, all agents broadcast their $k^{th}$ degree to their neighbours.
At the next round, every agent broadcasts all the $k^{th}$ degrees received (removing repetitions) to its neighbours so that in $2n$ rounds every agent has received a complete list of all the $k^{th}$ degrees of all nodes. Every agent $\alpha(v)$ orders all $k^{th}$ degrees (this can easily be done in such a way that all agents produce the same order), and sets $\delta_k(v)$ to be the rank of its own degree in the order. Then it proceeds to send out this new encoding of $\delta_k(v)$ and to calculate $\delta_{k+1}(v)$ accordingly. 

In this way, we have $s_{\max}=\mathcal{O}(\log(n+m))=\mathcal{O}(\log n)$. Note that the total number of rounds of this algorithm is $\mathcal{O}(n^2)$ and that, provided every set of degrees is stored as an ordered array, the cost of each computation done locally by an agent at a given round is bounded above by the size, $\mathcal{O}((n+m)s_{\max})=\mathcal{O}(ms_{\max})$, of the largest message sent.
\end{proof}

\Asafe*

\begin{proof}
Let $\nu$ be any solution in $I$ satisfying $\nu(y)=\nu(z)$ for every $y,z\in X$ with $y \sim_{\delta} z$ and let $p=x_0c_0\dots,x_\ell$ be any walk in $S$ witnessing that $(x,d)$ is not $S$-supported, (i.e, $p$ is such that $x_0=x$, $x_0 \sim_{\delta} x_\ell$, and $d\not\in \{d\}+_{S} p$). Since $S$ is safe we have that $\nu(y)\in S_y$ for every $y\in X$. It remains to see that $\nu(x)\neq d$, so that the safety condition remains unaltered when $(x,d)$ is removed. First, it follows easily by induction that for every $1\leq i\leq \ell$, $\nu(x_i)\in\{\nu(x)\}+_S p_i$ where $p_i=x_0c_0\dots x_i$. Then, since $\nu(x_\ell)\in\{\nu(x)\}+_S p$, $\nu(x)=\nu(x_\ell)$, and $d \notin \{d\} +_{S} p$, it follows that $\nu(x)\neq d$.
\end{proof}

\end{document}